\documentclass[prx,reprint,groupedaddress,nofootinbib,longbibliography,superscriptaddress]{revtex4-2}
\usepackage{amsmath,amssymb,graphicx,amsmath,amssymb,amsthm,bm,bbm,mathtools,amsfonts,array,empheq,colortbl}
\usepackage{xcolor}
\usepackage{phfqit}
\usepackage[normalem]{ulem}
\usepackage[unicode=true,pdfusetitle, bookmarks=true,bookmarksnumbered=false,bookmarksopen=false, breaklinks=false,pdfborder={0 0 0},backref=false,colorlinks=true, linkcolor=blue,citecolor=blue,urlcolor=blue]{hyperref}

\usepackage[caption=false]{subfig}
\usepackage{tikz}
\usetikzlibrary{positioning, shadings,calc}
\usetikzlibrary{arrows.meta, shapes.misc}

\definecolor{c4}{HTML}{2A52BE}
\definecolor{c1}{HTML}{de8001 }
\definecolor{c3}{HTML}{03b1d0}
\definecolor{c2}{HTML}{0087BD}

\newtheorem{theorem}{Theorem}

\newtheorem{lemma}{Lemma}

\newtheorem{remark}{Remark}

\newcommand{\Lc}{\mathcal{L}}

\newcommand{\Cc}{\mathcal{C}}

\newcommand{\Pe}{\hat{P}_E}

\newcommand{\id}{\mathrm{id}}
\newcommand{\gE}{g(E)}
\newcommand{\E}{\mathbb{E}}
\newcommand{\Vc}{\mathcal{V}}\newcommand{\Wc}{\mathcal{W}}
\newcommand{\Qe}{\hat{Q}_E}

\newcommand{\wt}{\widetilde}

\newcommand{\tnorm}[1]{\left|\left|#1\right|\right|_1}

\newcommand{\mc}[1]{\mathcal{#1}}

\newcommand{\Bx}{\pmb{x}}
\newcommand{\By}{\pmb{y}}
\DeclareMathOperator{\e}{e}
\newcommand{\hyp}{\mathsf{Hyp}_{P}^{(N,k)}(Q)}
\newcommand{\mult}[1]{\mathsf{Mult}_{#1}^{(N,k)}(Q)}
\renewcommand{\spec}{\mathsf{Spec}}

\newcommand{\dbeta}{\mathrm{d}\beta}

\begin{document}

\title{Symmetry-driven thermalization via finite de Finetti theorems}
\author{Uttam Singh}
\email{uttam@iiit.ac.in}
\affiliation{Centre for Quantum Science and Technology (CQST), International Institute of Information Technology Hyderabad, Gachibowli 500032, Telangana, India}
\affiliation{Center for Security, Theory and Algorithmic Research (CSTAR), International Institute of Information Technology Hyderabad, Gachibowli 500032, Telangana, India}

\author{Nicolas J. Cerf\,}  
\email{nicolas.cerf@ulb.be}
\affiliation{Centre for Quantum Information and Communication, \'{E}cole polytechnique de Bruxelles,  CP 165, Universit\'{e} libre de Bruxelles, 1050 Brussels, Belgium}

\begin{abstract}
Thermal behavior in subsystems of closed quantum systems is commonly attributed to dynamical chaos, quantum ergodicity, canonical typicality, or 
the eigenstate thermalization hypothesis, suggesting a fundamentally statistical origin of thermalization. Here, we propose a potential alternative mechanism in which thermal structures emerge deterministically from symmetry considerations alone, without recourse to statistical arguments. We prove a finite de Finetti-type theorem for quantum states invariant under energy-preserving unitaries, establishing that the reduced 
marginals of any such invariant $N$-qudit state are close (both in trace distance and relative entropy) to convex mixtures of thermal product states, with explicit error bounds vanishing as $N \to \infty$. We further present an example of energy-conserving Lindblad dynamics whose long-time limit is invariant under energy-preserving unitaries, providing a dynamical realization of the desired symmetry class. These results imply that invariance under energy-preserving unitaries suffices as a sole fundamental, deterministic principle to enforce thermal structures.
\end{abstract}
\maketitle

\noindent
{\it Introduction.} Experiments across a wide range of physical situations have explored thermalization in isolated quantum systems, often demonstrating robust local thermal behavior but also exhibiting sometimes striking departures from it in integrable regimes \cite{Kinoshita2006, Langen2015a, Langen2015b, Kaufman2016}. While many systems exhibit relaxation of small subsystems toward thermal (or generalized thermal) states despite a globally unitary, energy-conserving dynamics, others may fail to do so because of strong dynamical constraints. These observations have motivated a large body of theoretical work aimed at clarifying how thermal behavior can emerge, or be obstructed, under reversible microscopic laws \cite{Jaynes1957a, Jaynes1957b, Pusz1978, Lenard78, Popescu2006, Goldstein2006, Deutsch1991, Srednicki1994, Rigol2008, Reimann2008, Rigol2009, Linden2009, Short2012, DAlessio2016, Abanin2019, Mori2018}. Taken together, these approaches can be understood as addressing three interconnected fundamental problems: (P1) identifying the equilibrium thermal states; (P2) explaining how reduced states of subsystems may possess the thermal structure identified in (P1); and (P3) finding physical scenarios in which the microscopic dynamics gives rise to the behavior of (P2), thereby leading to thermalization. For example, contributions addressing (P1) include approaches based on static or variational characterizations of equilibrium states \cite{Jaynes1957a, Jaynes1957b, Pusz1978, Lenard78}. The body of work under (P2) notably includes canonical typicality and the eigenstate thermalization hypothesis \cite{Popescu2006, Goldstein2006, Deutsch1991, Srednicki1994, Rigol2008}. Results falling under (P3) focus on equilibration and relaxation under unitary time evolution \cite{Reimann2008, Linden2009, Short2012, Kaufman2016}.

 The central objective of this work is focused on (P2). Starting from the observation that all current approaches are intrinsically probabilistic, we question whether this is fundamentally unavoidable. Canonical typicality relates thermal behavior to typical properties of states within an energy shell, whereas the eigenstate thermalization hypothesis attributes it to generic features of energy eigenstates and the associated energy spectrum. In both cases, thermalization is inferred from typicality, which is a probabilistic or statistical notion, rather than from symmetry principles or structural constraints on the allowed unitary dynamics. Accordingly, these statistical approaches are formulated without explicit reference to restrictions such as integrability, conserved quantities, or symmetry-imposed selection rules \cite{DAlessio2016, Abanin2019}.

Here, we investigate whether the constraints imposed on the admissible dynamics can serve as a starting point, thereby circumventing the need for a statistical model. In the considered scenario, energy conservation provides a minimal and physically natural restriction, confining the evolution to unitaries that commute with the Hamiltonian. Because this constraint does not single out a unique unitary, and because resolving the detailed transformation would, in practice, require precise microscopic control, we instead describe the dynamics at an effective level by averaging over all energy-preserving unitaries, i.e., by twirling with respect to the group of energy-preserving unitaries (EPU). The relevant states are then those invariant under this group.

Guided by this operational viewpoint, we characterize the states of an arbitrary quantum system that are invariant under all EPUs. For any such EPU-invariant state, we prove that the reduced state of a fixed-size subsystem is close (both in trace distance and relative entropy) to a convex mixture of thermal product states, with explicit finite-size bounds. Furthermore, if the total energy distribution is sufficiently narrow, this mixture collapses to an effective single Gibbs state, hence the subsystem thermalizes. Energy-conservation symmetry therefore fixes the local structure of EPU-invariant states and enforces local thermal behavior.

\begin{figure}[htbp!]
    \centering
    \includegraphics[width=\linewidth]{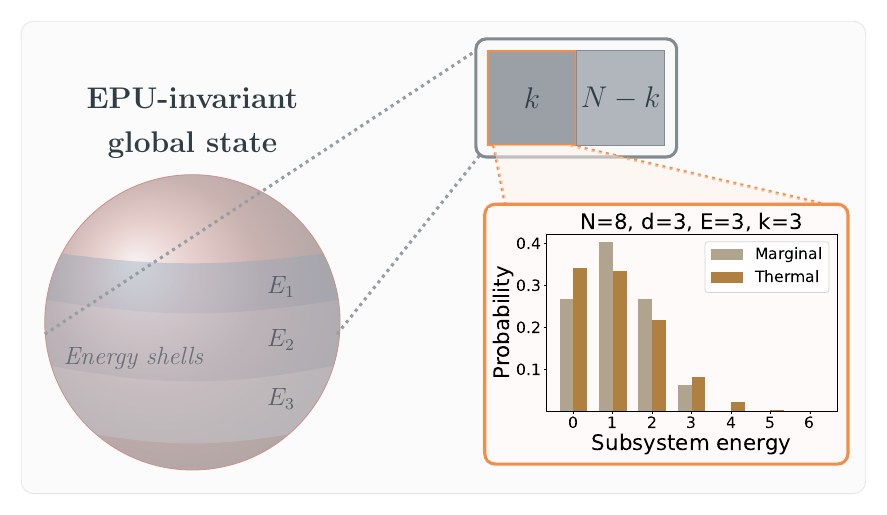}
    \caption{Symmetry-driven mechanism for the origin of thermalization. The schematic depicts the finite-dimensional de Finetti theorem established here, showing that invariance under EPUs constrains the $k$-qudit subsystems of a $N$-qudit system to exhibit thermal reduced states (with a nearly constant temperature if the total energy is narrowly distributed). The inset shows the example of a marginal energy distribution for qutrits ($d=3$) with equidistant energy levels in the case $N=8$, $k=3$, and $E=3$.}
\label{fig:sche3}
\end{figure}

The connection between global symmetry considerations and thermal marginals is illustrated schematically in Fig.~\ref{fig:sche3} for a subsystem of three out of eight qutrits. This mechanism provides a structural explanation of (P2) in that it yields a symmetry-based deterministic account of thermalization at the level of reduced states, without recourse to randomness, chaotic dynamics, or assumptions about typicality or generic spectral properties. Although this constitutes a purely static model for the existence of thermal states in subsystems (P2), we also describe a dynamical scenario in which EPU invariance naturally emerges in the long-time limit, thereby providing a possible route toward thermalization (P3).

Formally, our main result is a de Finetti–type theorem for finite-dimensional quantum states that are invariant under all EPUs. This result is inspired by a related de Finetti theorem for continuous-variable systems \cite{Leverrier2009}, where invariance under orthogonal symplectic transformations in phase space (i.e., particle-number preserving Gaussian operations in state space) was shown to lead to convex mixtures of Gaussian thermal states. In the present setting, our theorem yields a quantitative characterization of reduced states at finite system size and even remains valid under approximate invariance, with deviations being controlled by a suitable measure of asymmetry. At the heart of our approach, we make use of the method of types \cite{Cover2006, Csiszar2011}, a powerful tool of large deviation theory which constitutes a refinement of the asymptotic equipartition theorem. We note that a special case of thermalization in qubit systems with translation invariance has been analyzed in Ref. \cite{Muller2015}, which can be understood as exhibiting a finite de Finetti–type structure.

\medskip
\noindent
{\it Mathematical preliminaries.} Consider $N$ qudits ($d$-level quantum systems), each with a Hamiltonian $h$, and denote as $\mc{D}_{N}$ the set of all density matrices of $N$ qudits. Let the eigenvalues (assumed nondegenerate and equidistant) of $h$  be given by $E_x$, with the corresponding eigenstates $\{\ket{x}\}$ for $x\in \mathcal{X}:=\left\{0,\cdots,d-1\right\}$. Let the total Hamiltonian of the $N$ qudits be $H=\sum_{k=1}^N h^{(k)}$. 
The total energy eigenbasis is therefore $\{ \ket{\Bx} := \ket{x_1,\ldots,x_N} \}_{\Bx\in \{0,\ldots,d-1\}^N}$.
A unitary $U$ acting on the full space is energy-preserving if $[U,H]=0$. Let $\mc{U}_H$ be the set of $d^N\times d^N$ unitary matrices $U$ satisfying $[U,H]=0$. A state $\rho$ that is invariant under all unitaries in $\mc{U}_H$ (called EPU-invariant) must be block-diagonal in the energy eigenbasis with equal weight on all basis vectors sharing the same total energy. Let $\mc{S}_H$ be the set of states that are invariant under all unitaries in $\mc{U}_H$, i.e., $\mc{S}_H := \left\{\rho\in\mc{D}_N\,|\,U\rho U^\dag =\rho\, ,\forall\,U\in\mc{U}_H\right\}$. This set is a convex and compact subset of $\mc{D}_N$ (see Appendix~\ref{subsec:convex-compact-set}), so we will focus on its extreme points, namely the normalized projectors onto the energy eigenspaces of all possible total energies $E$.
A special example of state in $\mc{S}_H$ is given by the $N$-fold tensor product of single qudit thermal states $\tau_\beta=e^{-\beta\, h}/Z_{\beta}$, where $\beta$ is inverse temperature and $Z_\beta=\tr\left(e^{-\beta\, h}\right)$ is the partition function. For all $U\in\mc{U}_H$, it follows that $U\tau_\beta^{\otimes N}U^\dagger = \tau_\beta^{\otimes N}$. 

\medskip
\noindent
{\it Method of types.} This method allows us to parametrize energy shells in terms of empirical distributions (called types), providing a bridge between combinatorics and thermodynamics. Given a sequence $\Bx=x_1\dots x_N$ of length $N$, the \emph{type} $P_{\Bx}$ of sequence~$\Bx$ is the relative frequencies of all symbols of $\mathcal{X}$ within sequence $\Bx$ (it can be understood as the empirical probability distribution inferred from sequence $\Bx$, see, e.g., Refs. \cite{Cover2006, Csiszar2011}). That is, if $n_x$ is the number of occurrences of symbol $x$ in sequence $\Bx$, then $P_{\Bx}(x)=n_x/N$, with $x\in \mathcal{X}$.  Let $\mathcal{P}_N$ denote the set of all possible types $P$ of length-$N$ sequences \footnote{For example, if $\mathcal{X}=\{0,1\}$,
$d=2$, and $N=3$, we have 
$\mathcal{P}_N =\{(1,0),(2/3,1/3),(1/3,2/3),(0,1)\}$.}.  Then, for $P\in \mathcal{P}_N$, one defines the \emph{type class} $T_P$ of $P$ as 
\begin{align*}
 T_P:= \left\{\Bx\in \mathcal{X}^N \,|\,  P_{\Bx}=P\right\},   
\end{align*}
that is, the set of all length-$N$ sequences $\Bx$ whose empirical probability distribution matches $P$. Since the energy of sequence $\Bx$ is given by $\sum_{x=0}^{d-1}n_x \, E_x$, all sequences $\Bx\in T_P$ have the same total energy given by $E(P)=N\sum_{x=0}^{d-1} E_x\, P(x)$. The size of the type class $T_P$ is given by a multinomial coefficient, namely
\begin{align}
    |T_P|&=\binom{N}{NP}=\frac{N!}{\prod_{j=0}^{d-1}\big(N P(j)\big)!}.
\end{align}
Since there may be, in general, more than a single type satisfying the same total energy constraint, let us define the set of types with energy $E\in\spec(H)$ as 
\begin{align}
 \mc{S}(E):= \left\{P\in \mathcal{P}_N\,|\, E(P)=E\right\}.   
\end{align}
Physically, $\mc{S}(E)$ accounts for distinct occupation patterns (types) associated with the same total energy $E$, and the corresponding type classes $T_P$ enumerate the underlying basis vectors at that energy.  
The total number of strings with energy $E$ is given by $g(E) := \sum_{P\in \mc{S}(E)} |T_P|$.

The extreme states $\rho_{E}^{(N)}$, labeled by total energy $E$, are the normalized projectors onto the energy eigenspaces, i.e., the subspaces spanned by all basis states $\ket{\Bx}$ whose total energy is $E$. Equivalently, these are the strings $\Bx\in \mathcal{X}^N$ whose type $P_{\Bx}$ lies in $\mc{S}(E)$. Thus,
\begin{align}
\label{eq:extremalstate}
\rho_{E}^{(N)} 
&=\sum_{P\in \mc{S}(E)}\mu_E(P)\times \frac{1}{|T_P|}\sum_{\Bx\in T_P}\ket{\Bx}\bra{\Bx},
\end{align}
where $\mu_E(P)=|T_P| / g(E)$
is the fraction of basis states of total energy $E$ that have type $P$.



\noindent
{\it Accessing marginals.} 
To express our results, we need to evaluate the partial traces of $N$-qudit extremal states. Consider a basis sequence $\By=y_1\dots y_k$ of length $k$. The type $Q_{\By}$ of sequence $\By$ is the relative occurrences of the elements of $\mathcal{X}$ within $\By$. That is, if $m_y$ is the number of occurrences of symbol $y\in \mathcal{X}$ in the sequence $\By$, then $Q_{\By}(y)=m_y/k$. For $Q\in \mathcal{P}_k$, define the type class of $Q$ as
$T_Q :=\left\{\By\in \mathcal{X}^k\,|\, Q_{\By}=Q\right\}$, with size $|T_Q|=\binom{k}{kQ}$.

Partial tracing the last $N-k$ qudits of $\rho_E^{(N)}$ as given by Eq. \eqref{eq:extremalstate}, we get the following $k$-qudit reduced state
\begin{align}
\label{eq:ptracestate-step1}
\rho_{E}^{(k)} := \sum_{P\in \mc{S}(E)}\mu_E(P)\times \frac{1}{|T_P|}\sum_{\Bx\in T_P}\tr_{N-k}\left(\ket{\Bx}\bra{\Bx}\right).
\end{align}
For each type $P\in \mc{S}(E)$, we may express
\begin{align}
\label{eq:ptracestep1}
&\frac{1}{|T_P|}\sum_{\Bx\in T_P}\tr_{N-k}\left(\ket{\Bx}\bra{\Bx}\right)\nonumber\\
&=\!\!\!\sum_{\substack{Q\in \mathcal{P}_k\text{~such that}\\kQ(j)\leq NP(j), \, \forall j\in \mathcal{X}}} \!\!\!\frac{\prod_{j=0}^{d-1}\binom{NP(j)}{kQ(j)}}{\binom{N}{k}} \times \frac{1}{|T_Q|}\sum_{\By\in T_Q}\ket{\By}\bra{\By}.
\end{align}
Defining the multivariate hypergeometric distribution
\begin{align}
\hyp=\frac{\prod_{j=0}^{d-1}\binom{NP(j)}{kQ(j)}}{\binom{N}{k}}
\end{align}
and defining the normalized  projector onto type-$Q$ length-$k$ sequences as $\hat{\Pi}_Q:=\frac{1}{|T_Q|}\sum_{\By\in T_Q}\ket{\By}\bra{\By}$, we have
\begin{align}
\label{eq:ptracestate5}
\rho_{E}^{(k)} =\sum_{Q\in \mathcal{P}_k}\left(\sum_{P\in \mc{S}(E)} \mu_E(P)\,\,\hyp\right) \hat{\Pi}_Q.
\end{align} 
Note that the above condition \(kQ(j) \le NP(j)\) in the sum over $Q\in \mathcal{P}_k$ may be dropped \footnote{The condition \(kQ(j) \le NP(j)\) for all \(j\in \mathcal{X}\)
arises from the requirement that the factors \(\binom{NP(j)}{kQ(j)}\) in Eq. \eqref{eq:ptracestep1} must be nonzero. When we sum over \(P\), this condition is automatically enforced by the fact that the corresponding term vanishes whenever it is violated, so we may formally drop the condition.}.

We need to compare $\rho_{E}^{(k)}$ to the $k$th product of the one-qudit energy-diagonal state $\eta_P:=\sum_{x=0}^{d-1}P(x)\ket{x}\bra{x}$, for any type $P\in \mc{S}(E)$. It is easy to see that
\begin{align}
\label{eq:thermaldecomp2}
\eta_P^{\otimes k}
&=\sum_{Q\in\mathcal{P}_k}\mult{P} \,\,\hat{\Pi}_Q,
\end{align}
where we have defined the multinomial distribution as
\begin{align}
\mult{P}=\binom{k}{kQ}
\prod_{i=0}^{d-1} P(i)^{kQ(i)}.
\end{align}
The fact that $\eta_P^{\otimes k}$ depends on $N$ comes from $P\in \mc{S}(E)$. Indeed, $\eta_P$ has the same energy $E/N$ for all $P\in \mc{S}(E)$.

\medskip
\noindent
{\it Main result.} We are now ready to prove a finite de Finetti theorem for EPU-invariant states, which implies that thermal structures emerge deterministically in subsystems of an EPU-invariant state.

\begin{theorem}
\label{th:trace-de-finetti}
Let $\rho^{(N)}$ be an EPU-invariant $N$-qudit state and $\tau_{\beta}^{\otimes k}$ be the thermal (Gibbs) state at inverse temperature $\beta$ for $k$ non-interacting qudits.
Then, for any $k$-qudit marginal of $\rho^{(N)}$, there exists a mixture of \mbox{$k$-qudit} thermal states $\{\tau_{\beta}^{\otimes k}\}_{\beta}$ such that, for large $N$ and $k\ll \sqrt{N}$, we have
\begin{align}
\label{eq:t-norm-bound}
&\tnorm{\tr_{N-k}\left(\rho^{(N)}\right) - \int\mu(\dbeta)\,\tau_\beta^{\otimes k}}\nonumber\\
&~~~~~~\leq \frac{k(dk + 7d-2k+2)}{2N}+  O\left(N^{-3/2+c\,\delta}\right),
\end{align}
where $\mu(\dbeta)$ is a valid probability measure on the set $\{\beta\}$ of inverse temperatures, $c>0$ is some $d$-dependent constant, and $\delta\in\left(0,1/2c\right)$.
\end{theorem}

\begin{proof}
We may write $\rho^{(N)}=\sum_{E\in \mathsf{Spec}(H)} c_E \, \rho^{(N)}_E$, i.e., a convex mixture of extreme states in $\mc{S}_H$, with the assumption that $E\ne NE_0$ and $NE_{d-1}$ \footnote{The energies $E=NE_0$ and $NE_{d-1}$ correspond to two trivial extreme states $\rho^{(N)}_E$ that are pure and remain pure after partial trace. Thus, we must exclude them from Theorem \ref{th:trace-de-finetti}  and assume $E\notin \{NE_0,NE_{d-1}\}$ in the rest of the proof.}. 
Let $\mu(\dbeta):=\sum_{E\in\spec(H)}c_E\,\delta (\beta-\beta_{E/N})\,\dbeta$ be a probability measure, where $\beta_{E/N}$ is the inverse temperature obtained by solving $\sum_{i=0}^{d-1}E_i\e^{-\beta_{E/N}E_i}=(E/N)\sum_{i=0}^{d-1}\e^{-\beta_{E/N}E_i}$. 
Then,
\begin{align}
   &\tnorm{\tr_{N-k}\left( \rho^{(N)}\right)-\int_{\beta}\mu(\dbeta)\,\tau_{\beta}^{\otimes k}}\nonumber\\ &= \tnorm{\sum_{E\in \mathsf{Spec}(H)} c_E  \left(\tr_{N-k}\left( \rho^{(N)}_E\right)- \tau{\scriptstyle(E/N)}^{\otimes k}\right)}\nonumber\\
   &\leq \sum_{E\in \mathsf{Spec}(H)} c_E \, \tnorm{ \rho^{(k)}_E- \tau{\scriptstyle(E/N)}^{\otimes k}},
\end{align}
where $\tau{\scriptstyle(E/N)}:=\tau_{\beta_{E/N}}$. 
Further, we use the triangle inequality $\tnorm{ \rho^{(k)}_E- \tau{\scriptstyle(E/N)}^{\otimes k}} \leq \mathsf{norm1} + \mathsf{norm2}$, where $\mathsf{norm1}=\tnorm{\rho_E^{(k)} -\sum_{P\in \mc{S}(E)}\mu_E(P)\,\eta_P^{\otimes k}}$ and $\mathsf{norm2}=\tnorm{\sum_{P\in \mc{S}(E)}\mu_E(P)\,\eta_P^{\otimes k}- \tau{\scriptstyle(E/N)}^{\otimes k}}$. Using Eqs. \eqref{eq:ptracestate5} and \eqref{eq:thermaldecomp2}, as well as the fact that $\hat{\Pi}_Q$ are orthogonal normalized projectors, we have
\begin{align}
\label{eq-bound-on norm1}
&\mathsf{norm1}\nonumber\\
&~~~\leq\sum_{P\in \mc{S}(E)} \mu_E(P)\sum_{Q\in\mathcal{P}_k}\left|\hyp-\mult{P}\right|\nonumber\\
     &~~~\leq \frac{2kd}{N},
\end{align}
where the second inequality comes from a bound on the total variation distance connecting sampling  \emph{with} and \emph{without} replacement. In our notations, it is proven that $\sum_{Q\in\mathcal{P}_k}\left|\hyp-\mult{P}\right|\leq \frac{2kd}{N}$ in Ref. \cite{Diaconis1980}. This bound is tight in the regime of interest $k\ll N$. 

Now, we upper bound $\mathsf{norm2}$ recalling that  $\eta_P\!=\!\sum_{x=0}^{d-1}P(x)\ket{x}\bra{x}$ and $\tau{\scriptstyle(E/N)}\!=\!\sum_{x=0}^{d-1}P_\beta(x)\ket{x}\bra{x}$, where $P_\beta(x) = e^{-\beta_{E/N}E_x}/\left(\sum_{x=0}^{d-1}e^{-\beta_{E/N}E_x}\right)$. We obtain
\begin{align}
\label{eq-bound-on norm2}
\mathsf{norm2}
&\leq  \frac{k(dk+3d-2k+2)}{2N} + O\left(N^{-3/2+c\,\delta}\right),
\end{align}
see Appendix \ref{subsec:main-results} for the technical proof of this inequality. Combining Eqs. \eqref{eq-bound-on norm1} and \eqref{eq-bound-on norm2} completes the proof of the theorem.
\end{proof}

\noindent
{\it Interpretation of Theorem \ref{th:trace-de-finetti}.} This finite de Finetti-type theorem implies that the $k$-qudit marginals of an EPU-invariant $N$-qudit state are uniformly well approximated, for large $N$ and $k\ll \sqrt{N}$, by convex mixtures of thermal product states at different inverse temperatures. In Appendix \ref{subsec:simple}, this convergence is illustrated in the simple case of qutrits ($d=3$) with equidistant energy levels. Note that the approximation does not, in general, lead to a single Gibbs state. An EPU-invariant state may have support on several energy shells, each with a distinct effective temperature, so that the marginal inherits a mixture weighted by the energy distribution. Equivalently, the subsystem is thermal but at an unknown inverse temperature $\beta$ with $\mu(\mathrm{d}\beta)$ expressing the uncertainty about $\beta$. When the energy distribution is narrow, as for macroscopic systems, $\mu$ concentrates and the mixture approaches a single Gibbs state, yielding conventional thermal behavior in subsystems. In this sense, the finite de Finetti theorem identifies a symmetry-driven mechanism for the emergence of thermal structures, while retaining quantitative control over finite-size deviations. In particular, for any fixed subsystem size $k$, the deviation between the reduced state and the corresponding thermal mixture vanishes as the total system size $N$ increases, with a convergence rate scaling as $1/N$.

Although the proof employs asymptotic techniques such as Stirling’s approximation and Laplace’s method, all estimates are non-asymptotic and hold uniformly for fixed N, $k\ll \sqrt{N}$, and $d$. In addition, we strengthen this result in Appendix \ref{sup:rel-de-finneti} by proving an analogous de Finetti theorem where the relative entropy is used instead of the trace distance to bound the distance between the $k$-qudit marginals of an EPU-invariant state and the mixture of $k$-qudit thermal states. Furthermore, we also consider scenarios where the states satisfy only approximate invariance under EPUs. In those cases, we show that the distance between the $k$-qudit marginals of an EPU-invariant state and the mixture of $k$-qudit thermal states is additionally controlled by a specific measure of asymmetry of the total quantum state.

\noindent
{\it Dynamical emergence of EPU-invariant states.} Why should EPU-invariant states arise in physical scenarios at all? To answer this question, we consider a physically motivated class of open-system dynamics that converges to an exactly EPU-invariant state, thereby enabling a direct application of our de Finetti theorem. 
We take our system of $N$ qudits with Hamiltonian $H$, and, for each total energy $E$, we define the projector 
$\hat P_E:=\sum_{P\in S(E)}\sum_{\Bx\in T_P}\ket{\Bx}\bra{\Bx}$. Note that $\mathrm{tr}(\hat P_E)=g(E)$.
Let the initial state of the system be given by $\rho_0$.  We consider the dynamics generated by a Lindbladian of the form $\mathcal{L} := \mathcal{L}_{\mathrm{block}} + \mathcal{L}_{\mathrm{deph}}$, where
\begin{align}
\!\mathcal{L}_{\mathrm{block}}(\rho_0)
&=
\sum_{E} \gamma_E
\left[
\frac{\mathrm{tr}(\hat P_E \rho_0 \hat P_E)}{g(E)}\, \hat P_E - \hat P_E \rho_0 \hat P_E
\right],
\label{eq:Lblock-rig}
\\
\!\mathcal{L}_{\mathrm{deph}}(\rho_0)
&=
\sum_{E} \lambda_E 
\left[
\hat P_E \rho_0 \hat P_E - \tfrac{1}{2}\{\hat P_E,\rho_0\}
\right],
\label{eq:Ldeph-rig}
\end{align}
with $\lambda_E\ge 2\gamma_E \ge0$. Note that the Hamiltonian $H$ does not affect the long-time behavior, so we only consider the energy-preserving dissipative dynamics that is induced by $\mathcal{L}$. We show in Appendix \ref{eq:Lindblad} that such a Lindbladian arises in a collision model, and that solving the resulting dynamics yields
\begin{align}
    \rho(t) 
    &:= e^{t\mathcal{L}}(\rho_0)\nonumber\\
    &=\sum_{E}\left[(1-\varepsilon_E(t))\,\hat P_E\rho_0\hat P_E+\varepsilon_E(t)\,\mathrm{tr}(\hat P_E\rho_0)\,\frac{\hat P_E}{g(E)}\right]
\notag\\
&\quad +
\sum_{E\neq E'} 
e^{-\frac{1}{2}(\lambda_E+\lambda_{E'})t}\,
\hat P_E\rho_0 \hat P_{E'},
\label{eq:rho-t-full}
\end{align}
where $\varepsilon_E(t)=1-e^{-t\gamma_E}$. As $t\to\infty$, $\varepsilon_E(t)\to1$ and $e^{-\frac{1}{2}(\lambda_E+\lambda_{E'})t} \to 0$, yielding the limiting state
\begin{align}
\lim_{t\to\infty}\rho(t)=\sum_E \mathrm{tr}(\hat P_E\rho_0)\,\frac{\hat P_E}{g(E)}=:\mathcal{T}_{\mathrm{EPU}}(\rho_0).
\end{align}
Thus, the dissipative dynamics converges exponentially to the EPU-twirl of the initial state, which is the unique EPU-invariant fixed point of the semigroup.
 
Let $\Delta\rho_0 = \rho_0-\mathcal{T}_{\mathrm{EPU}}(\rho_0)$. Since $e^{t\mc{L}}\left(\mathcal{T}_{\mathrm{EPU}}(\rho_0)\right)= \mathcal{T}_{\mathrm{EPU}}(\rho_0)$, we have $e^{t\mc{L}}\left(\Delta\rho_0\right) = \rho(t)-\mathcal{T}_{\mathrm{EPU}}(\rho_0)$. Thus, $\tnorm{e^{t\mc{L}}(\rho_0)-\mathcal{T}_{\mathrm{EPU}}} =\tnorm{e^{t\mc{L}}\left(\Delta\rho_0\right)}$. Then, we show in Appendix \ref{eq:Lindblad} that 
\begin{align}
\tnorm{e^{t\mc{L}}(\rho_0)-\mathcal{T}_{\mathrm{EPU}}} \leq 2 e^{-t\Gamma}\tnorm{\Delta\rho_0}
   \leq 4 e^{-t\Gamma},    
\end{align}
where $\Gamma:=\min_E\{\gamma_E, \lambda_E\}$ and $\tnorm{\Delta\rho_0}\le 2$. Choosing $t\geq \Gamma^{-1}\log(4/\epsilon)$ implies that $\tnorm{e^{t\mc{L}}(\rho_0)-\mathcal{T}_{\mathrm{EPU}}}\leq \epsilon$. This means that any state approaches its EPU-twirl under the above Lindbladian dynamics on a timescale
\(O\left(\Gamma^{-1}\log(1/\epsilon)\right)\). Therefore, as a consequence of Theorem~\ref{th:trace-de-finetti}, its marginals thermalize on a timescale
\(O\left(\Gamma^{-1}\log(1/\epsilon)\right)\).

\medskip
\noindent
{\it Conclusion.} We have established a symmetry-based, deterministic mechanism for the emergence of thermal behaviors in quantum subsystems. By proving finite de Finetti–type theorems for states invariant under energy-preserving unitaries, we showed that thermal marginals follow directly from symmetry, with explicit finite-size and dimension-dependent bounds. This yields a non-statistical account of thermalization that does not rely on chaos, randomness, or ensemble sampling, and hence contrasts with established frameworks such as the the eigenstate thermalization hypothesis or canonical typicality. We further showed that this symmetry-based route to thermalization can arise as long-time limits of physically motivated Lindbladian dynamics.

Our proof uses the non-interacting structure $H=\sum_k h^{(k)}$, which allows the energy shells to be identified with type classes. With interactions, two cases arise according to whether the local terms commute. If they do, then e.g. for nearest-neighbour couplings, shells become classes of second-order types and remain exponentially degenerate~\cite{Csiszar1998}. A de Finetti theorem may then be expected \cite{Zaman1986}, with the thermal product states replaced by the reduction of the global Gibbs state. If they do not, the degeneracies are usually lifted, and the symmetry itself would have to be coarse grained to restore them \cite{Muller2015}. It will be very interesting to explore these directions.

Finally, since diagnosing thermalization directly from dynamics is difficult \cite{Devulapalli2025}, it is natural to seek simpler criteria. Noting that symmetry can often be tested efficiently \cite{LaBorde2023, Rethinasamy2025}, the symmetry-based approach developed here suggests that thermal structures could be inferred from the observed EPU-invariance. More broadly, our analysis points to symmetry as a potential operational witness of thermality in constrained quantum systems.

\medskip
\noindent
{\it Acknowledgments.} US thanks S. Das, V. Pandey, S. Jay, A. Ghorui, and L. Zhang for insightful discussions. US acknowledges support from the Ministry of Electronics and Information Technology (MeitY), Government of India, under Grant No. 4(3)/2024-ITEA. US thanks IIIT Hyderabad for support via the Faculty Seed Grant. NJC acknowledges support from the Fonds de la Recherche Scientifique (F.R.S.–FNRS) under Grant No T.0060.26 as well as under project CHEQS within the Excellence of Science (EOS) program.

\bibliography{de-Finetti}

\appendix
\onecolumngrid

\ \\ \ \\ \ \\
\noindent
The Appendix is organized as follows. In Appendix~\ref{subsec:simple}, we introduce a simple model of $N$  qutrits and numerically demonstrate the convergence of marginals of EPU–invariant states toward thermal states. In Appendix~\ref{subsec:convex-compact-set}, we prove that the set of EPU-invariant states is a convex and compact set.  Appendix~\ref{subsec:main-results} contains the proof of a key technical ingredient underlying Theorem~\ref{th:trace-de-finetti}. In Appendix~\ref{sup:rel-de-finneti}, we state and prove a second de Finetti theorem for EPU–invariant states, establishing convergence to mixtures of thermal states in terms of relative entropy. Finally, Appendix~\ref{eq:Lindblad} discusses approximate symmetry and presents the details of a Lindblad dynamics that drives a quantum system toward an EPU–invariant state.

\section{Simple example}
\label{subsec:simple}

Consider a system of $N$ qutrits ($d=3$), each with Hamiltonian $h = \ketbra{1}{1} + 2\, \ketbra{2}{2}$, and let $H = \sum_{i=1}^N h^{(i)}$ be the total Hamiltonian, where $h^{(i)}=h$ for all $1\leq i\leq N$. The states with total energy $1$ are those with exactly one system in state $\ket{1}$ and all other qutrits in state $\ket{0}$. There are exactly $N$ such configurations, denoted as
\begin{align}
\ket{e_i^{(1,N)}} := \ket{0}^{\otimes (i-1)} \otimes \ket{1} \otimes \ket{0}^{\otimes (N - i)}, ~~ 1\leq i\leq N.
\end{align}
Now consider the maximally mixed state over the subspace with total energy $1$, namely
\begin{align}
\rho^{(N)}_1 = \frac{1}{N} \sum_{i=1}^N \ketbra{e_i^{(1,N)}}{e_i^{(1,N)}},
\label{eq:A2}
\end{align}
which is obviously invariant under all EPUs. The reduced state on the first $k$ qutrits, obtained by tracing out the remaining $N-k$ qutrits, is given by
\begin{align}
\label{eq:E=1}
\rho^{(k)}_1 = \tr_{k+1,\dots,N} \left(\rho_1^{(N)}\right) = \frac{1}{N} \left((N - k) \ketbra{0}{0}^{\otimes k}+ \sum_{i=1}^k \ketbra{e_i^{(1,k)}}{e_i^{(1,k)}}  \right).
\end{align}
Similarly, consider now the states with total energy $2$. Such states fall into two distinct classes.

\medskip
\noindent
i) States with two qutrits in  $\ket{1}$ and the rest in $\ket{0}$. There are $\binom{N}{2}$ such states, of the form
  \begin{align}
  \ket{e^{(1,N)}_{i,j}} := \ket{0}^{\otimes (i-1)} \otimes \ket{1} \otimes \cdots \otimes \ket{1} \otimes \ket{0}^{\otimes (N-j)},~~1 \le i < j \le N.
  \end{align}

\begin{figure}[t!]
\centering
\subfloat[$k=1$ and $E=1$]{
{\includegraphics[width=0.45\textwidth]{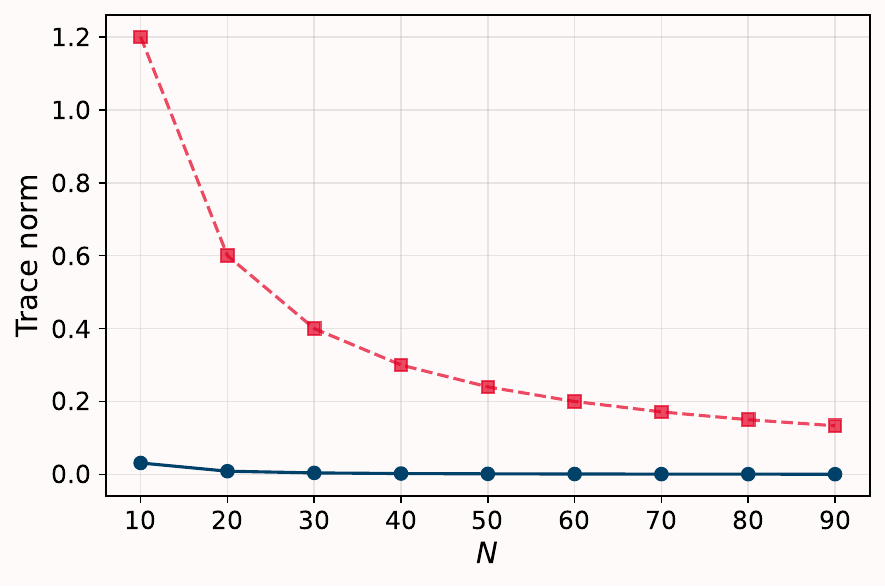}}\label{fg.N1}}
\subfloat[$k=1$ and $E=2$]{
{\includegraphics[width=0.45\textwidth]{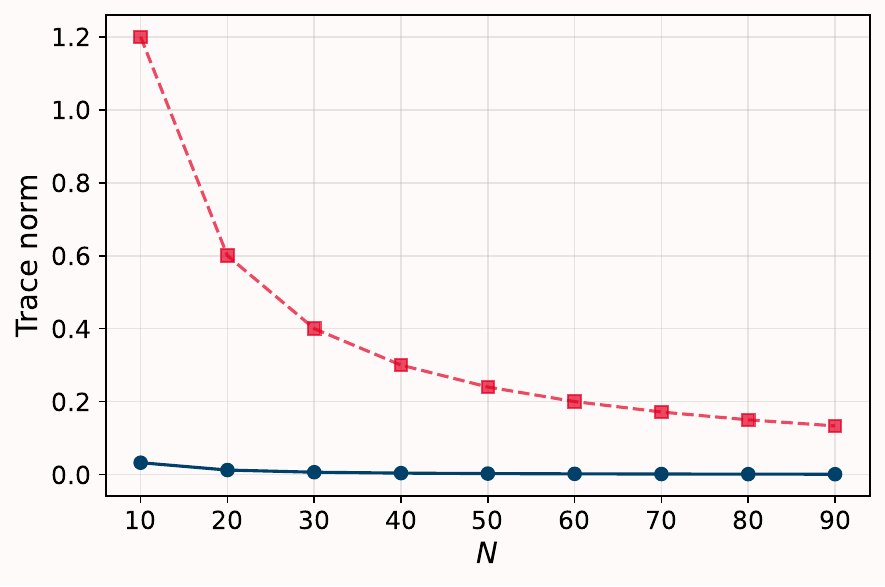}}\label{fg.N2}}
\\
\subfloat[$N=100$ and $E=1$]{
{\includegraphics[width=0.45\textwidth]{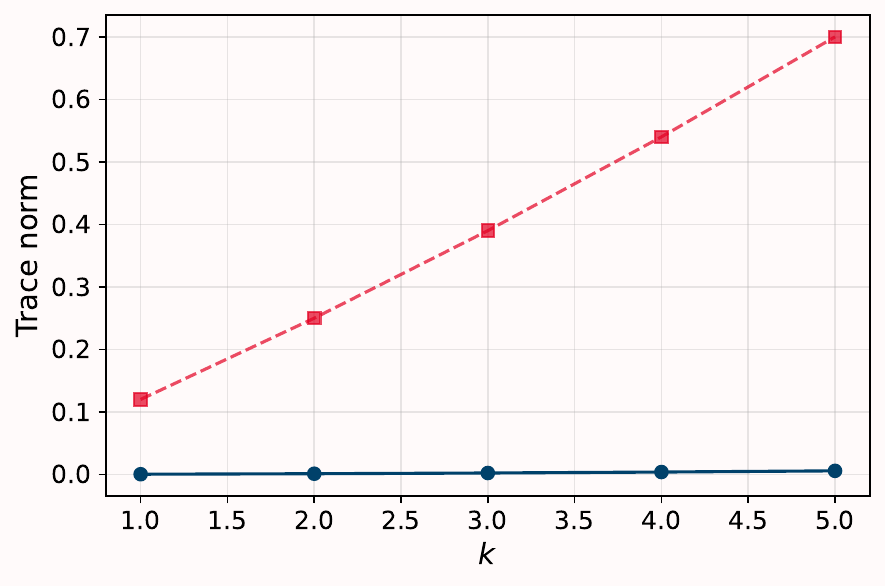}}\label{fg.k1}}
\subfloat[$N=100$ and $E=2$]{
{\includegraphics[width=0.45\textwidth]{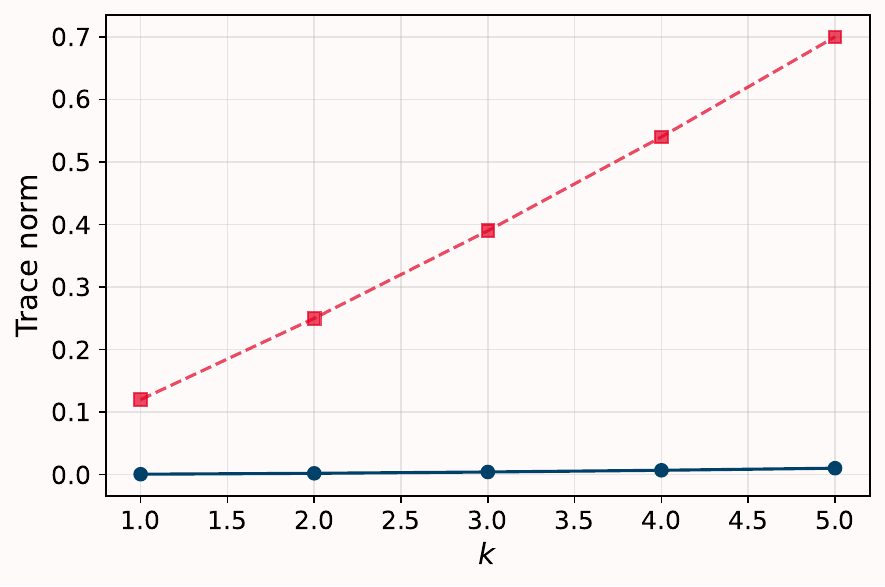}}\label{fg.k2}}
\caption{Convergence of reduced marginals to thermal mixtures. The single-qutrit Hamiltonian is taken to be $h=\ketbra{1}{1}+2\ketbra{2}{2}$ and the input states are uniform mixtures of energy $E=1$ and $E=2$ eigenstates of the total $N$-qutrit Hamiltonian. The $k$-qutrit marginals are given by Eqs. \eqref{eq:E=1} and \eqref{eq:E=2} for total energy $E=1$ and $E=2$, respectively, while the target thermal state is determined by the energy $E/N$ per qutrit. Panels (a) and (b) show the trace norm between the single qutrit marginals and the single qutrit thermal state as function of $N$ for fixed $E=1$ and $E=2$, respectively. Panels (c) and (d) show the trace norm between $k$-qutrit marginals of $N=100$ qutrit EPU-invariant state and $k$-qutrit thermal state as function of $k$ for $E=1$ and $E=2$, respectively. The red line in each curve represents the theoretical bound of Theorem \ref{th:trace-de-finetti}.
}
\label{fg.trace-de-finetti}
\end{figure}

\medskip
\noindent
ii) States with a single qutrit in $\ket{2}$ and the rest in $\ket{0}$. There are $N$ such states of the form
\begin{align}
  \ket{e^{(2,N)}_i} := \ket{0}^{\otimes (i-1)} \otimes \ket{2} \otimes \ket{0}^{\otimes (N-i)},~~1\leq i \leq N.
\end{align}
The total number of such energy-2 configurations is $M =\frac{N(N+1)}{2}$ and the uniform mixture of them is expressed as 
\begin{align}
\rho^{(N)}_2 = \frac{1}{M} \left( \sum_{1 \le i < j \le N} \ketbra{e^{(1,N)}_{i,j}}{e^{(1,N)}_{i,j}} + \sum_{i=1}^N \ketbra{e^{(2,N)}_i}{e^{(2,N)}_i} \right).
\label{eq:A6}
\end{align}
We again compute the reduced state $\rho^{(k)}_2 = \tr_{k+1,\dots,N} \left( \rho^{(N)}_2\right)$ on the first \( k \) qutrits. The contributions decompose as follows: (a) $\binom{k}{2}$ terms from $e^{(1,k)}_{i,j}$ with $1 \le i< j \le k$, yielding two qutrits in $\ket{1}$ in the first  $k$ qutrits; (b) $k(N-k)$ terms with exactly one of $i, j \le k$, contributing states with a single qutrit in $\ket{1}$ within the first $k$ qutrits; (c) $\binom{N-k}{2}$ terms with both $i, j > k$, which contribute $ \ketbra{0}{0}^{\otimes k}$; (d) $k$ terms from $ e^{(2,k)}_i$ with $i \le k$, contributing states with a single qutrit in $\ket{2}$ at position $i \le k$; (e) $N-k$ terms from $e^{(2,k)}_i$ with $i > k$, contributing again $\ketbra{0}{0}^{\otimes k}$. Putting all these together, we obtain
\begin{align}
\label{eq:E=2}
\rho^{(k)}_2 = \frac{1}{M} \left(
\sum_{1 \le i < j \le k} \ketbra{e^{(1,k)}_{i,j}}{e^{(1,k)}_{i,j}}
+ \sum_{i=1}^k (N-k) \ketbra{e^{(1,k)}_i}{e^{(1,k)}_i}
+ \sum_{i=1}^k \ketbra{e^{(2,k)}_i}{e^{(2,k)}_i}
+ \binom{N-k+1}{2} \ketbra{0}{0}^{\otimes k}
\right).
\end{align}
A similar construction could be made with increasing complexity for $E>2$, but this Appendix aims at illustrating our results in a simple case.

Theorem \ref{th:trace-de-finetti} provides a finite de Finetti theorem implying that the reduced marginals of states that are invariant under all EPUs converge to thermal mixtures.  Here, the EPU-invariant states of total energy $E=1$ or $2$ (i.e., uniform mixtures of energy $E=1$ or $E=2$ eigenstates of the total $N$-qutrit Hamiltonian) are given by Eqs. \eqref{eq:A2} and \eqref{eq:A6}, respectively.
Then, the $k$-qutrit marginals are given by Eqs. \eqref{eq:E=1} and \eqref{eq:E=2} for total energy $E=1$ and $E=2$, respectively, and the target single-qutrit thermal state is determined by the total energy $E/N$ per qutrit. Figs. \ref{fg.N1} and \ref{fg.N2} show the trace distance between the single-qutrit marginal of the EPU-invariant states and the single-qutrit thermal state as function of $N$ for fixed $E=1$ and $E=2$, respectively. Figs. \ref{fg.k1} and \ref{fg.k2} show the trace distance between $k$-qutrit marginals of $20$-qutrit EPU-invariant state as a function of $k$ for $E=1$ and $E=2$, respectively. From these examples, numerically,  it is clear that the reduced marginals converge to a thermal state as $N$ increases (though the convergence is slower for larger-$k$ marginals). This corroborates our main point that the thermal structure emerges deterministically and is rooted in symmetry rather than statistics.

\section{Set $\mc{S}_H$ of EPU-invariant states}
\label{subsec:convex-compact-set}

The set of states that are invariant under all unitaries in $\mc{U}_H$ is denoted as
\begin{align}
\label{eq:inv-states}
    \mc{S}_H := \left\{\rho\in\mc{D}_N\,|\,U\rho U^\dag =\rho\, ,\forall\,U\in\mc{U}_H\right\}.
\end{align}
Let us prove that $\mc{S}_H\subseteq \mc{D}_N$ is a convex and compact set. The convexity of $\mc{S}_H$ follows by noting that for any two states $\rho,\,\sigma\in\mc{S}_H$, we have $p\rho+(1-p)\sigma\in \mc{S}_H$ for all $p\in[0,1]$. Now, consider any convergent sequence $\left(\rho_m\right)\subseteq \mc{S}_H$. Positivity, unit trace, and the invariance condition $U\rho_mU^\dag=\rho_m$, $\forall m$, $\forall U\in\mc{U}_H$, all pass to the limit under norm convergence. Hence $\rho:=\lim_{m\to\infty}\rho_m$ also belongs to $\mc{S}_H$, showing that $\mc{S}_H$ is  closed. Since $\mc{S}_H$ is a closed  and bounded subset of a finite-dimensional space, it is compact. As a result, every EPU-invariant can be written as a convex sum of the extreme points of $\mc{S}_H$. This brings a major simplification as we will only need to focus on the extreme states in $\mc{S}_H$ (given that partial tracing is linear).

\section{Proof details of Theorem \ref{th:trace-de-finetti}}
\label{subsec:main-results}

Throughout, the single-particle levels are labelled in increasing order,
$E_0<E_1<\cdots<E_{d-1}$, and are assumed \emph{commensurate}: there exist
$\chi\in\mathbb{R}$, $\chi'>0$ and integers $q_0<\cdots<q_{d-1}$ with
\begin{align}
E_x=\chi+q_x\,\chi',\qquad x=0,\dots,d-1,
\end{align}
equivalently, all ratios $(E_i-E_j)/(E_k-E_l)$ of level differences are rational. Theorem \ref{th:trace-de-finetti} to the equidistant case $q_x=x$. Commensurability enters at exactly one point, in the proof of Lemma~\ref{le:lemma1:appendix}. It guarantees that the entries of the matrix $B$ of Eq.~\eqref{eq:commensurate-B} are rational, hence that $\Lambda=\{w\in\mathbb{R}^{d-2}:Bw\in\mathbb{Z}^d\}$ is a lattice of full rank $d-2$ and an invertible $J$ with $\Lambda=J\mathbb{Z}^{d-2}$ exists. Physically, it is the condition under which distinct types can carry the same total energy, so that $\mathcal{S}(E)$ contains many types and $\mu_E$ concentrates on the maximum-entropy point $P_\beta$ of $\Delta(E)$. The exponential degeneracy $g(E)$ of a shell does not itself require commensurability, since
a single type class $T_P$ already has multinomial size. For an incommensurate spectrum, $\mathcal{S}(E)$ generically reduces to a single type $P$, and the $k$-qudit marginal is
still asymptotically the product $\eta_P^{\otimes k}$; but $P$ need not be close to $P_\beta$, so the de Finetti structure survives while thermality fails.

We take $d\ge 3$ throughout. For $d=2$ the two constraints defining $\Delta(E)$ already determine $P$ uniquely, so $\mathcal{S}(E)$ and $\Delta(E)$ are singletons and $P_\beta$
coincides with the unique type; the subspace $T$ of Eq.~\eqref{eq:eq-T-append} is trivial, and $\tnorm{\sum_{P\in\mathcal{S}(E)}\mu_E(P)\eta_P^{\otimes k}-\tau(E/N)^{\otimes k}}=0$ identically, so Theorem~\ref{th:trace-de-finetti} reduces to the exchangeability bound of Eq.~\eqref{eq:replace-append-bound} \cite{Diaconis1980}.

Recall that $\eta_P=\sum_{i=0}^{d-1}P(i)\ket{E_i}\bra{E_i}$ and $\tau{\scriptstyle(E/N)}=\sum_{i=0}^{d-1}P_\beta(i)\ket{E_i}\bra{E_i}$, where 
\begin{align}
P_{\beta_{E/N}}(i) = \frac{e^{-\beta_{E/N}E_i}}{\sum_{j=0}^{d-1}e^{-\beta_{E/N}E_j}}
\end{align}
and $\beta_{E/N}$ is the inverse temperature obtained by solving the equation 
\begin{align}
    \sum_{i=0}^{d-1}E_i\e^{-\beta_{E/N}E_i}=(E/N)\sum_{i=0}^{d-1}\e^{-\beta_{E/N}E_i}.
\end{align}
For a fixed energy $E$, we will identify $\beta_{E/N}\equiv \beta$. Let
\begin{align}
\Delta(E):=\Big\{P\in\mathbb{R}^d_{\ge0}\ :\ \sum_{i=0}^{d-1}P(i)=1,\ \
\sum_{i=0}^{d-1}E_i\,P(i)=E/N\Big\}
\label{eq:energy-simplex}
\end{align}
denote the energy-constrained simplex. $P_{\beta}\in\Delta(E)$, but $P_{\beta}$ is in general
not a type, since $\mathcal{S}(E)$ consists of rational vectors $l/N$ with $l\in\mathbb{Z}^d_{\ge0}$. Since, for every $P\in \Delta(E)$, $-\sum_i P(i)\ln P_{\beta}(i)=(E/N)\beta+\ln \left(\sum_{i=0}^{d-1}\e^{-\beta E_i}\right)$ and $H(P_{\beta})=(E/N)\beta+\ln \left(\sum_{i=0}^{d-1}\e^{-\beta E_i}\right)$, we have $D(P\|P_{\beta})=-H(P)-\sum_iP(i)\ln P_{\beta}(i)=H(P_{\beta})-H(P)$. Then
nonnegativity of the relative entropy $D(P\|P_{\beta})$, with equality iff $P=P_{\beta}$ implies that $P_{\beta}$ maximizes the entropy over $\Delta(E)$. Using the Pinsker's inequality \cite{CoverThomas2006} and $\|x\|_1\ge\|x\|_2$, we get
\begin{align}
H(P_{\beta})-H(P)\  \ge\ \tfrac12\norm{P-P_{\beta}}_2^2 .
\end{align}

Set
\begin{align*}
E_{\mathrm{avg}}:=\sum_{i=0}^{d-1}P_\beta(i)E_i=\frac{E}{N},\quad
\sigma^2:=\sum_{i=0}^{d-1}P_\beta(i)(E_i-E_{\mathrm{avg}})^2,\quad
\bar{\varepsilon}:=\frac1d\sum_{i=0}^{d-1}E_i,\quad
s:=\min_{0\le i\le d-2}\left(E_{i+1}-E_i\right).
\end{align*}
Note that  $\left(\bar{\varepsilon}-E_0\right)+\left(E_{d-1}-\bar{\varepsilon}\right)=E_{d-1}-E_0$ and since $E_0\leq \bar{\varepsilon}\leq E_{d-1}$, we have $\max\{\bar{\varepsilon}-E_0,E_{d-1}-\bar{\varepsilon}\}\geq (E_{d-1}-E_0)/2$. Define the spectral ratio $\kappa$ as
\begin{align}
\label{eq:kappa-def}
\kappa:=\frac{\max\left\{\bar{\varepsilon}-E_0,\;E_{d-1}-\bar{\varepsilon}\right\}}{s}.
\end{align}
Moreover, since $E_{d-1}-E_0=\sum_{i=0}^{d-2}\left(E_{i+1}-E_i\right)\geq(d-1)s$, we have $\kappa\geq(d-1)/2$. Equality forces both $E_{d-1}-E_0=(d-1)s$, i.e.\ all gaps equal, and $\bar{\varepsilon}=(E_0+E_{d-1})/2$, both of which hold precisely for an equidistant spectrum.

We need to upper bound $\mathsf{norm2}=\tnorm{\sum_{P\in \mc{S}(E)} \mu_E(P)\,\eta_P^{\otimes k}-
\tau{\scriptstyle(E/N)}^{\otimes k}}$. But before that let us recall the Poisson summation formula.
For a sufficiently regular function $f:\mathbb R^n\to\mathbb C$, the Poisson summation formula over a lattice reads as
\begin{align}
\sum_{\pmb{k}\in\mathbb Z^n-\pmb{a}} f(\pmb{k})=\sum_{\pmb{m}\in\mathbb Z^n}
e^{-2\pi i \pmb{m}\cdot \pmb{a}}\,\hat f(2\pi \pmb{m}),
\end{align}
where $\hat f(\xi)=\int_{\mathbb R^n} f(\pmb{x})e^{-i\xi\cdot \pmb{x}}\,d\pmb{x}$. In our setting, we will require $n=d-2$ and $\pmb{a}=\pmb{k}_\beta$ (see proof of  Lemma \ref{le:lemma1:appendix}), where $\pmb{k}\in\mathbb{Z}^{d-2}$. Let us define $i$th row of some matrix $BJ$ (see proof of  Lemma \ref{le:lemma1:appendix}) as $b_i^T$, i.e. $b_i^T=e_i^T BJ$, where $e_i$ is vector with all zeros except at $i$th position and $1$ at $i$th position. Consider following functions 
\begin{align*}
    &g_1(\pmb{k})=\e^{-\frac{1}{2N}\pmb{k}^T\wt{K} \pmb{k}},\,\,\,g_2(\pmb{k})=\e^{-\frac{1}{2N}\pmb{k}^T\wt{K} \pmb{k}}(b_i^T{\pmb k}),\,\,\,g_3(\pmb{k})=\e^{-\frac{1}{2N}\pmb{k}^T\wt{K} \pmb{k}}(b_i^T{\pmb k})(b_j^T{\pmb k}),\\
    &\text{and}~~ g_4(\pmb{k})=\e^{-\frac{1}{2N}\pmb{k}^T\wt{K} \pmb{k}}(b_i^T{\pmb k})(b_j^T{\pmb k}) (b_l^T{\pmb k}) (b_m^T{\pmb k}).
\end{align*}
where $\wt{K}$ is a $(d-2)\times (d-2)$ positive definite matrix. Then from the Poisson summation formula, for each $i=1,2,3,4$, we have 
\begin{align}
\sum_{\pmb{k}\in\mathbb Z^{d-2}-\pmb{k}_\beta} g_i(\pmb{k})=\sum_{\pmb{m}\in\mathbb Z^{d-2}}
e^{-2\pi i \pmb{m}\cdot \pmb{k}_\beta}\,\hat g_i(2\pi \pmb{m}).
\end{align}
A direct computation gives the value of $\hat g_1(\xi)$ as
\begin{align*}
\hat g_1(\xi)=\int_{\mathbb R^{d-2}}
\e^{-\frac{1}{2N}\pmb{x}^T\wt{K} \pmb{x}}\e^{-i\xi\cdot\pmb{x}}\,d\pmb{x}= \sqrt{\frac{(2\pi N)^{d-2}}{\det\wt{K}}}
\exp\!\left(-\frac{N}{2}\xi^T\wt{K}^{-1}\xi\right).
\end{align*}
Further $\hat g_2(\xi) = i\sum_{a=0}^{d-3}(b_i^T)_a \frac{\partial}{\partial\xi_a}\hat{g}_1(\xi)$, $\hat g_3(\xi) = -\sum_{a,a'=0}^{d-3} (b_i^T)_a (b_j^T)_{a'}\frac{\partial^2}{\partial\xi_a\xi_{a'}}\hat{g}_1(\xi)$, and
\begin{align*}
 \hat g_4(\xi)=\sum_{a,b,c,e}(b_i^T)_a(b_j^T)_b(b_l^T)_c(b_m^T)_e\,\frac{\partial^4}{\partial\xi_a\partial\xi_b\partial\xi_c\partial\xi_e}\hat g_1(\xi).   
\end{align*}
For $\pmb{m}\neq 0$, we see that, for each $i=1,2,3,4$, $\hat g_i(2\pi\pmb{m})$ behaves as $O\left(\e^{-cN}\right)$ for some positive constant $c$. That is only $m=0$ or $\xi=0$ terms contribute. Using the fact that $\wt{K}$ is a positive definite matrix, we have 
\begin{align}
\label{eq:zeroth}
&\sum_{\pmb{k}\in\mathbb Z^{d-2}-\pmb{k}_\beta} \e^{-\frac{1}{2N}\pmb{k}^T\wt{K} \pmb{k}}=\sqrt{\frac{(2\pi N)^{d-2}}{\det\wt{K}}}+O\left(\e^{-cN}\right),\\
&\sum_{\pmb{k}\in\mathbb Z^{d-2}-\pmb{k}_\beta} \e^{-\frac{1}{2N}\pmb{k}^T\wt{K} \pmb{k}} (b_i^T\pmb{k})=0+O\left(\e^{-cN}\right), \label{eq:first} \\
&\sum_{\pmb{k}\in\mathbb Z^{d-2}-\pmb{k}_\beta} \e^{-\frac{1}{2N}\pmb{k}^T\wt{K} \pmb{k}} (b_i^T\pmb{k})(b_j^T\pmb{k})=N\sqrt{\frac{(2\pi N)^{d-2}}{\det\wt{K}}}
\,\mathsf{Mat}_{ij}+O\left(\e^{-cN}\right),\label{eq:second}\\
&\sum_{\pmb{k}\in\mathbb Z^{d-2}-\pmb{k}_\beta}\e^{-\frac{1}{2N}\pmb{k}^T\wt{K}\pmb{k}}
(b_i^T\pmb k)(b_j^T\pmb k)(b_l^T\pmb k)(b_m^T\pmb k)
= N^2\sqrt{\frac{(2\pi N)^{d-2}}{\det\wt K}}
\left[\mathsf{Mat}_{ij}\mathsf{Mat}_{lm} + \mathsf{Mat}_{il}\mathsf{Mat}_{jm} \right.\nonumber\\
&\left.\qquad\qquad\qquad\qquad\qquad\qquad\qquad\qquad\qquad\qquad\qquad\qquad~~~~~+ \mathsf{Mat}_{im}\mathsf{Mat}_{jl} 
\right]
+ O\!\left(\e^{-cN}\right), \label{eq:fourth}
\end{align}
where $\mathsf{Mat}_{ij}:=b_i^T\wt{K}^{-1}b_j$.

Now we are ready to upper bound $\mathsf{norm2}=\tnorm{\sum_{P\in \mc{S}(E)} \mu_E(P)\,\eta_P^{\otimes k}-
\tau{\scriptstyle(E/N)}^{\otimes k}}$ as follows. Recall that for a system of $N$ qudits, $\hat P_E:=\sum_{P\in \mc{S}(E)}\sum_{\Bx\in T_P}\ket{\Bx}\bra{\Bx}$. For any state $\rho^{(N)}$ let
\begin{align}
\label{eq:pmin-append}
p_{\min}\left(\rho^{(N)}\right):=
\min_{\substack{E\in\mathrm{Spec}(H)\\ \hat P_E\rho^{(N)}\hat P_E\neq 0}}\ \
\min_{0\le j\le d-1} P_{\beta_{E/N}}(j),
\end{align}
where $P_{\beta_{E/N}}$ is the thermal distribution at the inverse temperature $\beta_{E/N}$ determined by the equation $\sum_{i=0}^{d-1}E_i\e^{-\beta_{E/N}E_i}=(E/N)\sum_{i=0}^{d-1}\e^{-\beta_{E/N}E_i}$.  We write $p_{\min}\left(\rho^{(N)}\right)$ simply as $p_{\min}$, and treat it as a fixed constant. All implied constants below, and the thresholds $N_{0}$ beyond which the estimates
hold, depend only on $\{E_i\}_{i=0}^{d-1}$, on $p_{\min}$ and on $\delta$. They do not depend on $E$ or $N$.

\begin{remark}
\label{re:zero-energy}
The condition $p_{\min}\left(\rho^
{(N)}\right) > 0$
holds if and only if $E_0 < \frac{E}{N} < E_{d-1}$, i.e., when the total energy $E$ is not at the extremal values $E = N E_0$ or $E = N E_{d-1}$. Equivalently, this corresponds to finite temperature, for which the Gibbs
distribution has full support.
\end{remark}

\begin{lemma}
\label{le:lemma1:appendix}
Let $P_{\beta}\in \Delta(E)$ be the unique maximizer of the Shannon entropy $H(P)=-\sum_{i=0}^{d-1}P(i)\ln P(i)$ over $\Delta(E)$, and assume $\min_{0\le j\le d-1} P_\beta(j) \ge p_{\min} > 0$. Then
\begin{align}
   \tnorm{\sum_{P\in \mc{S}(E)} \mu_E(P)\,\eta_P^{\otimes k}-\tau{\scriptstyle(E/N)}^{\otimes k}}\leq \frac{k(8\kappa+d(k-1)-2k+6)}{2N}+ O\left(N^{-3/2+c\delta}\right),
\end{align}
where $c>0$ is some $d$-dependent constant and $\delta\in \left(0,1/2c\right)$. And for equidistant spectrum, $\kappa=(d-1)/2$,
\begin{align}
\tnorm{\sum_{P\in \mc{S}(E)} \mu_E(P)\,\eta_P^{\otimes k}-
\tau{\scriptstyle(E/N)}^{\otimes k}}
&\leq  \frac{k(d(k+3)-2(k-1))}{2N}+ O\left(N^{-3/2+c\delta}\right).
\end{align}
\end{lemma}

\begin{proof}
Let $L$ be the following set of nonnegative integer vectors
\begin{align}
  L=\left\{\pmb{l}\in\mathbb{Z}_{\geq 0}^d:\sum_{i=0}^{d-1}l_i=N, ~ \sum_{i=0}^{d-1}l_iE_i=E\right\}.
\end{align} 
Then $\mc{S}(E)=L/N$. For $P\in \mc{S}(E)$, we have $\eta_P = \sum_{i=0}^{d-1} P(i)\, |E_i\rangle\!\langle E_i|$. The $k$-fold tensor power of $\eta_P$ can be written as
\begin{align*}
\eta_P^{\otimes k}=\sum_{i_1,\dots,i_k}P(i_1)\cdots P(i_k)
\,|E_{i_1},\dots,E_{i_k}\rangle\langle E_{i_1},\dots,E_{i_k}|.
\end{align*}
Expanding $\eta_P^{\otimes k}$, we obtain $\eta_P^{\otimes k} = \tau{\scriptstyle(E/N)}^{\otimes k}+T_1(P)+T_2(P)+O\left(\norm{P-P_\beta}_2^3\right)$, where $T_1(P)=\sum_{r=1}^{k}\sum_{i=0}^{d-1}(P(i)-P_\beta(i))|E_i\rangle\!\langle E_i|_r\otimes
\eta_{P_\beta}^{\otimes (k-1)}$  and $T_2(P)=\sum_{1\le r<s\le k}\sum_{i,j=0}^{d-1}\,(P(i)-P_\beta(i))(P(j)-P_\beta(j))|E_i\rangle\!\langle E_i|_r
\otimes|E_j\rangle\!\langle E_j|_s\otimes\eta_{P_\beta}^{\otimes (k-2)}$. $\tau{\scriptstyle(E/N)}$ is the thermal state determined by the energy constraint $E$ and is an energy diagonal state with probability distribution given by $P_\beta\in \Delta(E)$.  Then
\begin{align}
\label{eq:all-terms-trace}
\tnorm{\sum_{P\in \mc{S}(E)} \mu_E(P)\,\eta_P^{\otimes k}-
\tau{\scriptstyle(E/N)}^{\otimes k}}\leq \tnorm{\sum_{P\in \mc{S}(E)} \mu_E(P)\,T_1(P)} + \tnorm{\sum_{P\in \mc{S}(E)} \mu_E(P)\,T_2(P)} + \text{higher order terms}.
\end{align}
In particular,
\begin{align}
\label{eq:T1}
&\tnorm{\sum_{P\in \mc{S}(E)} \mu_E(P)\,T_1(P)} \leq k\sum_{i=0}^{d-1}\left|\sum_{P\in \mc{S}(E)} \mu_E(P)\,(P(i)-P_\beta(i))\right|,\\
&\tnorm{\sum_{P\in \mc{S}(E)} \mu_E(P)\,T_2(P)} \leq \binom{k}{2}\sum_{i,j=0}^{d-1}\left|\sum_{P\in \mc{S}(E)} \mu_E(P)\,(P(i)-P_\beta(i))(P(j)-P_\beta(j))\right|.\label{eq:T2}
\end{align}

Now we will compute sums in Eqs. \eqref{eq:T1} and \eqref{eq:T2}. Recall that $\mu_E(P)=\binom{N}{NP}/\left(\sum_{P\in \mc{S}(E)\binom{N}{NP}}\right)$. Applying Robbins' refinement of Stirling's approximation \cite{Robbins1955}, 
\begin{align*}
N! = N^N \e^{-N} \sqrt{2\pi N} \e^{\theta_N}, \quad \text{where } \frac{1}{12N + 1} < \theta_N < \frac{1}{12N}, 
\end{align*}
to $\ln \binom{N}{NP}= \ln N! -\sum_i\ln (NP(i))!$, and using $P(j)>0$ for all $j$, we get
\begin{align*}
\ln \binom{N}{NP} 
&=N H(P)+ \ln \left( \prod_j P(j)^{-1/2} \right)-\frac{(d-1)}{2}\ln (2\pi N) + \theta_N - \sum_j \theta_{N P(j)}.
\end{align*}
Then
\begin{align*}
\mu_E(P) &= \frac{\e^{NH(P)}\e^{- \sum_j \theta_{N P(j)}} \prod_j P(j)^{-1/2}}{\sum_{Q\in \mathcal{S}(E)}\e^{NH(Q)}\e^{- \sum_j \theta_{N Q(j)}} \prod_j Q(j)^{-1/2} }.
\end{align*}
Using the fact that $\frac{1}{12N + 1} < \theta_N < \frac{1}{12N}$, we have 
\begin{align*}
\mu_E(P) 
&= \frac{\e^{NH(P)} \prod_j P(j)^{-1/2}}{\sum_{Q\in \mathcal{S}(E)}\e^{NH(Q)} \prod_j Q(j)^{-1/2} }\left(1+O(1/N)\right).
\end{align*}
Thus, computation of Eqs. \eqref{eq:T1} and \eqref{eq:T2} reduces to computing sums of the form $F_N=   \sum_{P\in \mc{S}(E)}\e^{NH(P)} F(P)\left(1+O(1/N)\right)$, where $F$ is some polynomial on $\mc{S}(E)$ of the form $(P(i)-P_\beta(i))^r(P(j)-P_\beta(j))^s$ with $r,s\geq 0$.  Define
\begin{align*}
    &\mc{B}_{\delta}:=\left\{\pmb{x}\in \mathbb{R}^d\,|\,\norm{\pmb{x} - P_\beta}_2\leq N^{-1/2+\delta}\right\},
\end{align*}
where $0<\delta<1/2$. For $P\in \mc{S}(E)$, 
\begin{align*}
   H(P_{\beta})-H(P)\geq \frac{1}{2}\,\norm{P-P_{\beta}}_2^2.
\end{align*}
That is, for $P\notin \mc{S}(E)\cap \mc{B}_{\delta}$, we have $H(P_{\beta})-H(P)\geq \frac{1}{2}\,N^{-1+2\delta}$. Then $\e^{NH(P)}\lesssim \,\e^{NH(P_\beta)}\e^{-N^{2\delta}/2}$. This means outside $\mc{S}(E)\cap \mc{B}_{\delta}$, the terms $\e^{NH(P)} \,F(P)$ are exponentially suppressed. Thus, 
\begin{align}
\label{eq:dropping-outside-eq}
&F_N=   \left(\sum_{P\in \mc{S}(E)\cap \mc{B}_{\delta}}\e^{NH(P)} F(P) + O\left(\e^{-N^{2\delta}/2}\right)\right)\left(1+O(1/N)\right).
\end{align}
In the following, we will drop the terms of order $O\left(\e^{-N^{2\delta}/2}\right)$.  Using Taylor's theorem with Lagrange form of the remainder, we have $H(P) = H(P_\beta) - M(P) + R_3(P)$ with $M(P) = \frac{1}{2} (P - P_\beta)^TK(P-P_\beta)$, $K=\diag\left(1/P_\beta(0),\cdots,1/P_\beta(d-1)\right)$, and 
\begin{align*}
R_3(P) = \frac{1}{6} \sum_j \frac{(P(j) - P_\beta(j))^3}{\left(P_\beta(j)+t(P(j)-P_{\beta}(j))\right)^2},
\end{align*}
where $t\in(0,1)$ and the Taylor expansion is restricted to the $(d-2)$-dimensional affine subspace. 
For $P\in\mc{S}(E)\cap\mc{B}_\delta$, we have $|P(j) - P_\beta(j)|\leq N^{-1/2+\delta}$ for all $j$.  Further, since $P_\beta(j)\geq p_{\min}>0$, we have $P_\beta(j)+ t(P(j)-P_\beta(j))\geq p_{\min}-N^{-1/2+\delta}$ for $P\in\mc{S}(E)\cap\mc{B}_\delta$. Then for $P\in\mc{S}(E)\cap\mc{B}_\delta$, we have
\begin{align*}
   &-M(P) \geq -\frac{N^{-1+2\delta}}{2\,p_{\min}}\quad\text{and}\quad |R_3(P)| \le  \frac{d\,N^{-3/2+3\delta}}{6 (p_{\min} - N^{-1/2+\delta})^2}.
\end{align*}
Thus, $NM(P)=O\left(N^{2\delta}\right)$ while  $N|R_3(P)|=O\left(N^{-1/2+3\delta}\right)$. Using this, we have
\begin{align}
&F_N=   \e^{NH(P_\beta)}\sum_{P\in \frac{L}{N}\cap \mc{B}_{\delta}}\e^{-\frac{N}{2}(P-P_\beta)^TK(P-P_\beta)} F(P) \left(1+O\left(N^{-1/2+3\delta}\right)\right).
\end{align}
For $P\in \frac{L}{N} \cap\mc{B}_\delta$, $\pmb{u}=NP-NP_\beta=\pmb{l}-NP_\beta$ satisfies $\norm{\pmb{u}}\leq N^{1/2+\delta}$, where $\pmb{l}\in \mathbb{Z}_{\geq 0}^d$. Thus, $|u_i|\leq N^{1/2+\delta}$ or $l_i$ satisfies $NP_{\beta}(i)-N^{1/2+\delta} \leq l_i\leq NP_{\beta}(i)+N^{1/2+\delta}$. But since $P_\beta(i)\ge p_{\min}>0$, there exists $N_{0}$, depending only on $p_{\min}$ and
$\delta$, such that for all $N\ge N_{0}$ we have $P_\beta(i)\ge N^{-1/2+\delta}$ for all $i$. Then $NP_{\beta}(i)-N^{1/2+\delta} \geq 0$ for all $i$. Thus inside $\mc{B}_\delta$, $\pmb{u}$ can be taken to be inside $\mathbb{Z}^d-NP_\beta$ instead of $\mathbb{Z}_{\geq 0}^d-NP_\beta$ together with the constraints that $\sum_{i=0}^{d-1}u_i=0$ and $\sum_{i=0}^{d-1}u_iE_i=0$. Thus, we have
\begin{align*}
F_N&=  \e^{NH(P_\beta)} \sum_{\substack{\pmb{u}\in \mathbb{Z}^d-NP_\beta\\ \sum_iu_i=0, \sum_i u_i E_i=0\\ \norm{\pmb{u}}\leq N^{1/2+\delta}}} \e^{-\frac{1}{2N}\pmb{u}^TK\pmb{u}} F\left(\frac{\pmb{u}}{N}+P_\beta\right) \left(1+O\left(N^{-1/2+3\delta}\right)\right).
\end{align*}
Now let us fix the last two coordinates of $\pmb{u}$, namely $u_{d-2}$ and $u_{d-1}$ using the constraints. That is, $u_{d-2}+u_{d-1}=-\sum_{i=0}^{d-3}u_i$ and $E_{d-2}u_{d-2}+E_{d-1}u_{d-1}=-\sum_{i=0}^{d-3}E_iu_i$. For $\pmb{w}\in \mathbb{R}^{d-2}$, let us define $\pmb{u}=B\pmb{w}$, where $B$ is a $d\times (d-2)$ given by
\begin{align}
\label{eq:commensurate-B}
    B=\begin{pmatrix}
        1 & 0 & \cdots & 0\\
        0 & 1 & \cdots & 0\\
        \vdots & \vdots & \ddots  & \vdots\\
        0 & 0 & \cdots& 1\\
        \frac{-E_{d-1} + E_0}{E_{d-1} - E_{d-2}} & \frac{-E_{d-1}  + E_1}{E_{d-1} - E_{d-2}}  &\cdots & \frac{-E_{d-1}  + E_{d-3}}{E_{d-1} - E_{d-2}} \\
        \frac{E_{d-2}  - E_0}{E_{d-1} - E_{d-2}} & \frac{E_{d-2} - E_1}{E_{d-1} - E_{d-2}}  & \cdots& \frac{E_{d-2}  - E_{d-3}}{E_{d-1} - E_{d-2}}
    \end{pmatrix}.
\end{align}
Let 
\begin{align}
\label{eq:eq-T-append}
    T=\left\{x\in\mathbb{R}^d: \sum_ix_i=0, \sum_i E_ix_i=0\right\}.
\end{align}
Then it is easy to see that $B\pmb{w} \in T$. Let
$\Lambda = \{\pmb {w} \in \mathbb{R}^{d-2} : B{\pmb w} \in \mathbb{Z}^{d}\}$ and consider
\begin{align*}
  W := \{\, {\pmb w} \in \mathbb{R}^{d-2} \ :\ B{\pmb w} + NP_\beta \in \mathbb{Z}^{d} \,\}.  
\end{align*}
That $W$ is nonempty can be seen as follows. Fix any ${\pmb l}_{0} \in L$, then ${\pmb l}_{0} - NP_\beta \in T = \mathrm{col}(B)$, and so $B {\pmb w}_{0} = {\pmb l}_{0} - NP_\beta$ admits a solution ${\pmb w}_{0}$, which is unique as $B$ is injective, and ${\pmb w}_{0} \in W$. If ${\pmb w}, {\pmb w}' \in W$ then $B({\pmb w} - {\pmb w}') \in \mathbb{Z}^{d}$, i.e. ${\pmb w} - {\pmb w}' \in \Lambda$.
Hence $W = {\pmb w}_{0} + \Lambda$, and setting ${\pmb w}_\beta := -{\pmb w}_{0}$, the constraint
${\pmb u} + NP_\beta \in \mathbb{Z}^{d}$ with ${\pmb u} = B{\pmb w} \in T$ is equivalent to ${\pmb w} \in \Lambda - {\pmb w}_\beta$. The particular choice of ${\pmb l}_{0}$ is immaterial: a different choice shifts ${\pmb w}_\beta$ by an
element of $\Lambda$, leaving the sum unchanged. Using these facts, we have
\begin{align*}
F_N&=  \e^{NH(P_\beta)} \sum_{\substack{\pmb{w}\in \Lambda - \pmb{w}_\beta\\ \norm{B\pmb{w}}\leq N^{1/2+\delta}}} \e^{-\frac{1}{2N}\pmb{w}^TK'\pmb{w}} F\left(\frac{B\pmb{w}}{N}+P_\beta\right) \left(1+O\left(N^{-1/2+3\delta}\right)\right),
\end{align*}
where $K'=B^TK B$. Let $J$ be an invertible matrix such that $\Lambda=J\mathbb{Z}^{d-2}$. Let $\pmb{k}=J^{-1}\pmb{w}$ and $\pmb{k}_\beta=J^{-1}\pmb{w}_\beta$. Then
\begin{align}
F_N&=  \e^{NH(P_\beta)} \sum_{\substack{\pmb{k}\in \mathbb{Z}^{d-2} - \pmb{k}_\beta\\ \norm{BJ\pmb{k}}\leq N^{1/2+\delta}}} \e^{-\frac{1}{2N}\pmb{k}^T\wt{K}\pmb{k}} F\left(\frac{BJ\pmb{k}}{N}+P_\beta\right)\left(1+O\left(N^{-1/2+3\delta}\right)\right),
\end{align}
where $\wt{K}=J^TB^TK BJ$. Since $BJ$ is not dependent on $N$ and is an injective map from $\mathbb{R}^{d-2}\to\mathbb{R}^d$, there exist constants
$c_1>0$ and $c_2> 0$ such that $c_1\|\pmb{k}\|\leq\|BJ\pmb{k}\|\le c_2\|\pmb{k}\|$. Thus the condition $\|BJ\pmb{k}\|\le N^{1/2+\delta}$ is equivalent to $\|\pmb{k}\|\le C N^{1/2+\delta}$ for some constant $C>0$. Because $\wt{K}$ is positive definite, there exists $c>0$ with ${\pmb k}^T\wt{K} {\pmb k}\ge c\|\pmb{k}\|^2$. Splitting the exponent, for $\|\pmb{k}\|>CN^{1/2+\delta}$ we have $e^{-\frac{1}{2N}{\pmb k}^T\wt{K}{\pmb k}}\le e^{-\frac{c}{4N}\|{\pmb k}\|^2}e^{-\frac{c}{4}C^2N^{2\delta}}$,
while $F(BJ{\pmb k}/N+P_\beta)$ is polynomially bounded in $\|{\pmb k}\|$. Since
$\sum_{\pmb k}\|\pmb{k}\|^{m}e^{-\frac{c}{4N}\|{\pmb k}\|^{2}}=O\bigl(N^{(d-2+m)/2}\bigr)$ for every fixed $m\ge0$,
the tail contributes $O\bigl(e^{-c'N^{2\delta}}\bigr)$ for some $c'>0$, and the restriction to the ball can be removed with exponentially small error. We then have 
\begin{align}
\label{eq:FN-final}
F_N&=  \e^{NH(P_\beta)} \sum_{\pmb{k}\in \mathbb{Z}^{d-2} - \pmb{k}_\beta} \e^{-\frac{1}{2N}\pmb{k}^T\wt{K}\pmb{k}} F\left(\frac{BJ\pmb{k}}{N}+P_\beta\right)\left(1+O\left(N^{-1/2+3\delta}\right)\right).
\end{align}

We first compute the leading contributions in the expression $\sum_{P\in \mc{S}(E)} \mu_E(P)\,(P(i)-P_\beta(i))$. We have
\begin{align*}
 \sum_{P\in \mc{S}(E)} \mu_E(P)\,(P(i)-P_\beta(i)) = \frac{\sum_{P\in \mc{S}(E)}\e^{NH(P)} \prod_j P(j)^{-1/2}(P(i)-P_\beta(i))}{\sum_{Q\in \mathcal{S}(E)}\e^{NH(Q)} \prod_j Q(j)^{-1/2} }\left(1+O\left(1/N\right)\right).
\end{align*}
A remark is in order before evaluating this ratio. In the reduction leading to Eq. \eqref{eq:FN-final}, the cubic Taylor remainder was absorbed into a uniform relative factor, $\e^{NR_3(P)}=1+O\!\left(N^{-1/2+3\delta}\right)$, and pulled outside the sum. This is legitimate whenever the Gaussian sum it multiplies is nonzero at leading order and then $NR_3$ does not contribute to the leading order. But when the Gaussian sum does vanish, $NR_3$ must be included. Since, inside $\mc{B}_\delta$, $|P(i)-P_\beta(i)|/P_{\beta}(i)\to0$ for large enough $N$, we can write $\prod_j P(j)^{-1/2} =\left(\prod_j P_\beta(j)^{-1/2} \right) \left(1-\frac{1}{2} \sum_j\frac{P(j)-P_\beta(j)}{P_\beta(j)}+O\left(\norm{\frac{(P-P_\beta)}{P_\beta}}_2^2\right)\right)$, and the weight carries in addition the factor $\e^{NR_3(P)}=1+\frac{N}{6}\sum_j\frac{(P(j)-P_\beta(j))^3}{P_\beta(j)^2}+O\!\left(N^{-1+6\delta}\right)$. Using this we obtain to the leading order
\begin{align*}
 &\sum_{P\in \mc{S}(E)} \mu_E(P)\,(P(i)-P_\beta(i)) \\
 &= \frac{\sum_{P\in \mc{S}(E)\cap\mc{B}_\delta}\e^{-\frac{N}{2}(P-P_\beta)^TK(P-P_\beta)} (P(i)-P_\beta(i))\left[1-\frac{1}{2} \sum_j\frac{P(j)-P_\beta(j)}{P_\beta(j)}+\frac{N}{6}\sum_j\frac{(P(j)-P_\beta(j))^3}{P_\beta(j)^2}\right]}{\sum_{Q\in \mathcal{S}(E)\cap\mc{B}_{\delta}}\e^{-\frac{N}{2}(Q-P_\beta)^TK(Q-P_\beta)} }+O\left(N^{-3/2+c\delta}\right).
\end{align*}
Using the coordinate transformation, $P\mapsto \pmb{k}\in\mathbb Z^{d-2}-\pmb{k}_\beta$, by noting $P=\pmb{u}/N+P_\beta = B\pmb{w}/N+P_\beta = BJ\pmb{k}/N+P_\beta$ and $\wt{K}=J^TB^TKBJ$, and using the fact that the linear term vanishes, we have
\begin{align}
\label{eq:chain-for-linear-terms}
 &\sum_{P\in \mc{S}(E)} \mu_E(P)\,(P(i)-P_\beta(i))\nonumber\\
 &= -\frac{1}{2N^2}\sum_jP_\beta(j)^{-1}\frac{\sum_{\pmb{k}} \e^{-\frac{1}{2N}\pmb{k}^T\wt{K} \pmb{k}} (b_i^T{\pmb k})(b_j^T{\pmb k})}{\sum_{\pmb{k}} \e^{-\frac{1}{2N}\pmb{k}^T\wt{K} \pmb{k}}}
 +\frac{1}{6N^3}\sum_jP_\beta(j)^{-2}\frac{\sum_{\pmb{k}} \e^{-\frac{1}{2N}\pmb{k}^T\wt{K} \pmb{k}} (b_i^T{\pmb k})(b_j^T{\pmb k})^3}{\sum_{\pmb{k}} \e^{-\frac{1}{2N}\pmb{k}^T\wt{K} \pmb{k}}}
 +O\left(N^{-3/2+c\delta}\right)\nonumber\\
 &= -\frac{1}{2N}\sum_jP_\beta(j)^{-1}\,\mathsf{Mat}_{ij}
 +\frac{1}{2N}\sum_jP_\beta(j)^{-2}\,\mathsf{Mat}_{ij}\mathsf{Mat}_{jj}
 +O\left(N^{-3/2+c\delta}\right)\nonumber\\
 &= -\frac{1}{2N}\sum_jP_\beta(j)^{-1}\left[1-P_\beta(j)^{-1}\,\mathsf{Mat}_{jj}\right]\mathsf{Mat}_{ij}
 +O\left(N^{-3/2+c\delta}\right),
\end{align}
where in the second line we used Eqs.~\eqref{eq:second} and  \eqref{eq:fourth}, and $\mathsf{Mat}_{ij}=b_i^T\wt{K}^{-1}b_j$. $c$ is a $d$-dependent constant and $\delta\in(0,1/2c)$. The remaining corrections are odd sums of degree $\ge 5$ or even sums entering at relative order $O(N^{-1})$, hence $O(N^{-3/2+c\delta})$. A careful analysis in Ref.~\cite{Lapinski2020} confirms that all such corrections are bounded by $O(N^{-3/2+c\delta})$, smaller than the leading $1/N$ term as long as $\delta\in(0,1/2c)$.

Similarly, using the coordinate transformation and Eq. \eqref{eq:second}, we obtain (to the leading order)
\begin{align}
\label{eq:chain-of-inequalities-quadratic}
 &\sum_{P\in \mc{S}(E)} \mu_E(P)\,(P(i)-P_\beta(i))(P(j)-P_\beta(j))\nonumber \\
 &= \frac{\sum_{P\in \mathcal{S}(E)\cap\mc{B}_{\delta} }\e^{-\frac{N}{2}(P-P_\beta)^TK(P-P_\beta)} \left((P(i)-P_\beta(i))(P(j)-P_\beta(j))\right)}{\sum_{Q\in \mathcal{S}(E)\cap\mc{B}_{\delta}}\e^{-\frac{N}{2}(Q-P_\beta)^TK(Q-P_\beta)} } +O\left(N^{-3/2+c\delta}\right)\nonumber \\
 &= \frac{1}{N^2}\frac{\sum_{\pmb{k}\in\mathbb Z^{d-2}-\pmb{k}_\beta} \e^{-\frac{1}{2N}\pmb{k}^T\wt{K} \pmb{k}} (b_i^T{\pmb k})(b_j^T{\pmb k})}{\sum_{\pmb{k}\in\mathbb Z^{d-2}-\pmb{k}_\beta} \e^{-\frac{1}{2N}\pmb{k}^T\wt{K} \pmb{k}}  } +O\left(N^{-3/2+c\delta}\right)\nonumber \\
 &= \frac{1}{N}\mathsf{Mat}_{ij} +O\left(N^{-3/2+c\delta}\right).
\end{align}
Combining everything, from Eqs. \eqref{eq:all-terms-trace}, \eqref{eq:T1}, and \eqref{eq:T2}, we conclude that
\begin{align}
\label{eq:int-last-bounds}
\tnorm{\sum_{P\in \mc{S}(E)} \mu_E(P)\,\eta_P^{\otimes k}-\tau{\scriptstyle(E/N)}^{\otimes k}}
&\leq  \frac{k}{2N}\sum_{i=0}^{d-1}\left|\sum_{j=0}^{d-1}P_\beta(j)^{-1}\!\left[1-P_\beta(j)^{-1}\mathsf{Mat}_{jj}\right]\!\mathsf{Mat}_{ij}\right|\nonumber\\
&+ \frac{k(k-1)}{2N}\sum_{i,j=0}^{d-1}\left|\mathsf{Mat}_{ij}\right|+O\left(N^{-3/2+c\delta}\right).
\end{align}

In Lemma \ref{lemma:append-restriction-to-T}, we show that $\sum_{i=0}^{d-1}\left|\sum_{j=0}^{d-1}P_\beta(j)^{-1}\!\left[1-P_\beta(j)^{-1}\mathsf{Mat}_{jj}\right]\!\mathsf{Mat}_{ij}\right|\leq 4+8\kappa$ and $\sum_{i,j=0}^{d-1}\left|\mathsf{Mat}_{ij}\right|\leq (d-2)$. Thus,
\begin{align*}
\tnorm{\sum_{P\in \mc{S}(E)} \mu_E(P)\,\eta_P^{\otimes k}-
\tau{\scriptstyle(E/N)}^{\otimes k}}&\leq \frac{2k(1+2\kappa)}{N}+\frac{k(k-1)(d-2)}{2N} + O\left(N^{-3/2+c\delta}\right)\\
&= \frac{k(8\kappa+d(k-1)-2k+6)}{2N}+ O\left(N^{-3/2+c\delta}\right).
\end{align*}
And for equidistant spectrum use $\kappa=(d-1)/2$. This concludes the proof.
\end{proof}

\begin{remark}[Dependence of error in terms of $k$ and $N$] 
\label{rem:worst-scaling}
Note that in Lemma \ref{le:lemma1:appendix}, the contributions to $\tnorm{\sum_{P\in \mc{S}(E)} \mu_E(P)\,\eta_P^{\otimes k}-
\tau{\scriptstyle(E/N)}^{\otimes k}}$ involve terms like $k^{i}\norm{P-P_{\beta}}^{i}$ for even integer $i>0$ and $k^{i}\norm{P-P_{\beta}}^{i+1}$ for odd integer $i>0$. They give the worst-case scaling as $\left(\frac{k^2}{N}\right)^{i/2}$, where $i>0$, implying that the deviation from thermality remains small as long as $k$ does not grow faster than $\sqrt{N}$. 
\end{remark}

We now bound the terms $\sum_{i=0}^{d-1}\left|\sum_{j=0}^{d-1}P_\beta(j)^{-1}\!\left[1-P_\beta(j)^{-1}\mathsf{Mat}_{jj}\right]\!\mathsf{Mat}_{ij}\right|$ and $\sum_{i,j=0}^{d-1}\left|\mathsf{Mat}_{ij}\right|$ appearing in Eq. \eqref{eq:int-last-bounds} of  Lemma \ref{le:lemma1:appendix}  as the following lemma. 

\begin{lemma}
\label{lemma:append-restriction-to-T} 
Let $A=BJ$, and $\wt{K}= J^TB^TK BJ$. Then $BJ\wt{K}^{-1}J^TB^T=A(A^TKA)^{-1}A^T$ is the restriction of $K^{-1}$ to the subspace $T$. Let $b_i^T$ be the $i$th row of matrix $A=BJ$ and $\mathsf{Mat}_{ij}=b_i^T\wt{K}^{-1}b_j$, then
\begin{align}
&\sum_{i,j=0}^{d-1}\left| \mathsf{Mat}_{ij}\right|\leq d-2,\\
&\sum_{i=0}^{d-1}\left|\sum_{j=0}^{d-1}P_\beta(j)^{-1}\!\left[1-P_\beta(j)^{-1}\mathsf{Mat}_{jj}\right]\!\mathsf{Mat}_{ij}\right|\leq 4+8\,\kappa.
\end{align}
\end{lemma}

\begin{proof}
Define the constraint matrix
\begin{align}
C =
\begin{pmatrix}
1 & 1 & \cdots & 1 \\
E_0 & E_1 & \cdots & E_{d-1}
\end{pmatrix}.
\end{align}
Then the constraint subspace $T=\left\{x\in\mathbb{R}^d: \sum_ix_i=0, \sum_i E_ix_i=0\right\}$ can be written as $\ker{C}$. The matrix $B\in\mathbb R^{d\times(d-2)}$ was constructed so that $\pmb{u} = B\pmb{w} \in T$. Equivalently, $CB = 0$. Thus every column of $B$ lies in $T$, and therefore the column space $\operatorname{col}(B)$ satisfies $\operatorname{col}(B) \subseteq T$. The first $d-2$ rows of $B$ form the identity matrix, so the columns of $B$ are linearly independent. Hence $\operatorname{rank}(B) = d-2$. Since $\dim(T)=d-2$ and $\operatorname{col}(B)\subseteq T$, we conclude $\operatorname{col}(B) = T$. Let $A=BJ$, where $J\in\mathbb R^{(d-2)\times(d-2)}$ is an invertible matrix. Multiplying by an invertible matrix does not change the column space, so $\operatorname{col}(A) = \operatorname{col}(B) = T$. Thus the columns of $A=BJ$ form a basis of the constraint subspace $T$.
 
Consider the weighted inner product $\langle x,y\rangle_K = x^T K y$, where $K$ is a positive definite matrix. The projection $\mathsf{Proj}_K x$ of a vector $x$ onto $T$ with respect to this inner product
is characterized by the orthogonality condition $\langle x - \mathsf{Proj}_K x, u \rangle_K = 0$
for all  $u \in T$. Since every $u \in T$ can be written as $u = Ay$ for some $y$, the condition becomes $\langle x - \mathsf{Proj}_K x, Ay \rangle_K = 0$ for all $y$. Substituting the definition of the inner product gives $(x - \mathsf{Proj}_K x)^T K Ay = 0$ for all $y$. Since this holds for all $y$, we obtain $A^T K (x - \mathsf{Proj}_K x) = 0$. Hence $A^T K x = A^T K \mathsf{Proj}_K x$. Now $\mathsf{Proj}_K x$ lies in the column space of $A$, so there exists a vector $y$ such that $\mathsf{Proj}_K x = Ay$. Thus $A^T K x = A^T K A\, y$. Since $A$ has full column rank and $K$ is positive definite,
the matrix $A^T K A$ is invertible, and therefore
\begin{align}
Ay=\mathsf{Proj}_K x = A (A^T K A)^{-1} A^T K x.
\end{align}
Since this holds for all $x$, the projection matrix is $\mathsf{Proj}_K = A \wt{K}^{-1} A^T K$, where $\wt{K}= A^TK A$. It is easy to see that $\mathsf{Proj}_K^2=\mathsf{Proj}_K$. Now
\begin{align}
\label{eq:proj-constrained}
 \mathsf{Proj}_K K^{-1} = A \wt{K}^{-1} A^T=\mathsf{Mat}.  
\end{align}
$\mathsf{Mat}$ is the restriction of $K^{-1}$ to the subspace $T$. That is, $\mathsf{Mat}$ first applies $K^{-1}$ and then projects onto $T$. Further, $0\leq \mathsf{Mat}\leq K^{-1}$. $\mathsf{Mat}\leq K^{-1}$ follows by noting that $K^{-1}-\mathsf{Mat}=(\mathbb{I}-\mathsf{Proj}_K)K^{-1}\geq 0$. $\mathsf{Mat}\geq 0$ can be proved as follows. $x^T \mathsf{Mat}\,x= x^TA\wt{K}^{-1} A^T x = y^T\wt{K}^{-1}y$ for some $y=A^Tx$, then since $\wt{K}^{-1}\geq 0$, we conclude that $x^T \mathsf{Mat}\,x\geq 0$ or $\mathsf{Mat}\geq 0$. Note that $\sum_{ij}b_i^T\wt{K}^{-1}b_j = \sum_{ij}\mathsf{Mat}_{ij}$.  Now, from Cauchy--Schwarz inequality, we have
\begin{align*}
 \sum_{ij}\left|b_i^T\wt{K}^{-1}b_j\right| =\sum_{ij}\left|\mathsf{Mat}_{ij}\right| \leq    \sum_{ij}\sqrt{\mathsf{Mat}_{ii}\mathsf{Mat}_{jj}}. 
\end{align*}
Recalling that $K=\diag\left(P_\beta(0)^{-1},\cdots,P_\beta(d-1)^{-1}\right)$, we write
\begin{align*}
 \sum_{ij}\left|b_i^T\wt{K}^{-1}b_j\right|  &\leq    \left(\sum_{i}\sqrt{\mathsf{Mat}_{ii}K_{ii}}\sqrt{P_\beta(i)}\right)^2\\
 &\leq \left(\sum_{i=0}^{d-1}\mathsf{Mat}_{ii}K_{ii}\right)\left(\sum_{i=0}^{d-1}P_\beta(i)\right)=\tr\left(\mathsf{Mat}\,K\right)=\tr(\mathsf{Proj}_K) \leq d-2,
\end{align*}
where we used $\tr(\mathsf{Proj}_K)=\dim T=d-2$.

Let $q_j:=1-P_\beta(j)^{-1}b_j^T\wt{K}^{-1}b_j=1-(\mathsf{Proj}_K)_{jj}$, and let
$v:=\left(P_\beta(0)^{-1}q_0,\cdots,P_\beta(d-1)^{-1}q_{d-1}\right)^T$, so that $v=Kq$ and
$\sum_{j=0}^{d-1}P_\beta(j)^{-1}q_j\,\mathsf{Mat}_{ij}=\left(\mathsf{Mat}\,v\right)_i=\left(\mathsf{Proj}_K\,q\right)_i$. Then
\begin{align*}
\sum_{i=0}^{d-1}\left|\sum_{j=0}^{d-1}P_\beta(j)^{-1}\!\left[1-P_\beta(j)^{-1}b_j^T\wt{K}^{-1}b_j\right]\!b_i^T\wt{K}^{-1}b_j\right|=\norm{\mathsf{Proj}_K\,q}_1.
\end{align*}
Note that since $\mathsf{Proj}_K^2=\mathsf{Proj}_K$, $0\le\mathsf{Proj}_K\le I$ and its diagonal entries satisfy $(\mathsf{Proj}_K)_{jj}\in[0,1]$, hence $q_j\in[0,1]$. And since $\tr(\mathsf{Proj}_K)=d-2$, we have $\sum_jq_j=d-\tr(\mathsf{Proj}_K)=2$. Let $y:=\mathsf{Proj}_K\,q$ and since $\mathsf{Proj}_K\,q\in T$, we have $\sum_iy_i=0$. Write $q=y+z$ with $z\perp_KT$. Since $z\perp_KT$, we have $Kz\in T^\perp=\text{col}(C^T)$ and we can write $z_i=P_\beta(i)(\alpha+\gamma E_i)$ for some $\alpha,\gamma\in\mathbb{R}$. Imposing $y=q-z\in T$, i.e.\ $\sum_iz_i=\sum_iq_i=2$ and $\sum_iE_iz_i=\sum_iE_iq_i$, gives $\alpha+\gamma E_{\mathrm{avg}}=2$ and $\alpha E_{\mathrm{avg}}+\gamma\sum_iP_\beta(i)E_i^2=\sum_iE_iq_i$. Subtracting $E_{\mathrm{avg}}$ times the first from the second yields $\gamma\sigma^2=\sum_iq_i(E_i-E_{\mathrm{avg}})$, and therefore $\gamma=\dfrac{\sum_iq_i(E_i-E_{\mathrm{avg}})}{\sigma^2}$ and
\begin{align}
\label{eq:f-explicit}
y_i=q_i-P_\beta(i)\left(2+\gamma\left(E_i-E_{\mathrm{avg}}\right)\right).
\end{align}
Since $\sum_iy_i=0$, and $(q_i-z_i)^+ \leq q_i + z_i^-$ for any $q_i\geq 0$ and $z_i\in \mathbb{R}$, we have
\begin{align*}
\norm{\mathsf{Proj}_K\,q}_1=\norm{y}_1\leq 2\sum_i(q_i+z_i^-)= 4 + 2\sum_i z_i^{-}.
\end{align*}
Recall that $z_i= P_\beta(i)\left(2+\gamma\left(E_i-E_{\mathrm{avg}}\right)\right)$. If $z_i<0$ then $\gamma P_\beta(i)(E_i-E_{\mathrm{avg}})<-2 P_\beta(i)<0$, then $\left(\gamma P_\beta(i)(E_i-E_{\mathrm{avg}})\right)^-=-\gamma P_\beta(i)(E_i-E_{\mathrm{avg}})\geq -2P_\beta(i)-\gamma P_\beta(i)(E_i-E_{\mathrm{avg}}) =z_i^-$. Thus, $z_i^- \leq P_\beta(i)\left(\gamma (E_i-E_{\mathrm{avg}})\right)^-$. For $z_i>0$, this is trivial. Then
\begin{align*}
\norm{\mathsf{Proj}_K\,q}_1&\leq 4 + 2\sum_iP_\beta(i)\left(\gamma (E_i-E_{\mathrm{avg}})\right)^-\\
&=4+ |\gamma |\sum_iP_\beta(i)\left|E_i-E_{\mathrm{avg}}\right|.
\end{align*}
Choose $j$ with $E_j\leq E_{\mathrm{avg}}\leq E_{j+1}$ and set $a=E_{\mathrm{avg}}-E_j\geq0$, $b=E_{j+1}-E_{\mathrm{avg}}\geq0$, so that $a+b=E_{j+1}-E_j\geq s$. For $i\leq j$ we have $E_{\mathrm{avg}}-E_i\geq a$, and for $i>j$ we have $E_i-E_{\mathrm{avg}}\geq b$. Therefore
\begin{align*}
\sigma^2&=\sum_{i\leq j}P_\beta(i)\left(E_{\mathrm{avg}}-E_i\right)^2+\sum_{i>j}P_\beta(i)\left(E_i-E_{\mathrm{avg}}\right)^2\\
&\geq a\sum_{i\leq j}P_\beta(i)\left|E_{\mathrm{avg}}-E_i\right|+b\sum_{i>j}P_\beta(i)\left|E_i-E_{\mathrm{avg}}\right|\\
&=\frac{a+b}{2}\sum_iP_\beta(i)\left|E_i-E_{\mathrm{avg}}\right|\geq\;\frac{s}{2}\sum_iP_\beta(i)\left|E_i-E_{\mathrm{avg}}\right|.
\end{align*}
That is, $\sum_iP_\beta(i)\left|E_i-E_{\mathrm{avg}}\right|\leq 2\sigma^2/s$. Using $q_i\in[0,1]$ and $\sum_iq_i=2$,
\begin{align*}
|\gamma|\,\sigma^2=\left|\sum_iq_i(E_i-E_{\mathrm{avg}})\right|\le\sum_iq_i\left|E_i-E_{\mathrm{avg}}\right|\le 2\max_i\left|E_i-E_{\mathrm{avg}}\right|.
\end{align*}
Since $E_0<E_{\mathrm{avg}}<E_{d-1}$, we have $\max_i|E_i-E_{\mathrm{avg}}|\le E_{d-1}-E_0$. Together with $\sum_iP_\beta(i)|E_i-E_{\mathrm{avg}}|\le 2\sigma^2/s$,
\begin{align*}
\left|\gamma\right|\sum_iP_\beta(i)\left|E_i-E_{\mathrm{avg}}\right|\leq \frac{2\sigma^2}{s}\cdot\frac{2\max_i|E_i-E_{\mathrm{avg}}|}{\sigma^2}=\frac{4(E_{d-1}-E_0)}{s}.
\end{align*}
Recalling $E_{d-1}-E_0\le 2\kappa s$, we get $|\gamma|\sum_iP_\beta(i)|E_i-E_{\mathrm{avg}}|\le 8\kappa$. Combining everything, we have $\norm{\mathsf{Proj}_K\,q}_1\le  4+8\kappa$. Thus,
\begin{align}
\sum_{i=0}^{d-1}\left|\sum_{j=0}^{d-1}P_\beta(j)^{-1}\!\left[1-P_\beta(j)^{-1}b_j^T\wt{K}^{-1}b_j\right]\!b_i^T\wt{K}^{-1}b_j\right|\leq 4+ 8\,\kappa.
\end{align}
This concludes the proof of the lemma.
\end{proof}

Let $\rho^{(N)}=\sum_{E}c_E\rho_E^{(N)}$ be an EPU-invariant $N$-qudit state with $p_{\min}(\rho^{(N)})>0$. Here $0\leq c_E\leq 1$ and $\sum_Ec_E=1$. Let $\mu(\dbeta)=\sum_{E}c_E\,\delta (\beta-\beta_{E/N})\,\dbeta$ be a probability measure. Then
\begin{align*}
  \tnorm{\tr_{N-k}\left(\rho^{(N)}\right) - \int\mu(\dbeta)\,\tau_\beta^{\otimes k}}\leq \sum_{E}c_E \tnorm{\tr_{N-k}\left(\rho_E^{(N)}\right) - \tau_{\beta_{E/N}}^{\otimes k}}.
\end{align*}
Applying triangle inequality, we have
\begin{align*}
\tnorm{\tr_{N-k}\left(\rho_E^{(N)}\right) - \tau_{\beta_{E/N}}^{\otimes k}}\leq \tnorm{\tr_{N-k}\left(\rho_E^{(N)}\right) -\sum_{P\in \mc{S}(E)}\mu_E(P)\,\eta_P^{\otimes k}}+ \tnorm{\sum_{P\in \mc{S}(E)}\mu_E(P)\,\eta_P^{\otimes k} - \tau_{\beta_{E/N}}^{\otimes k}}.
\end{align*}
Now using sampling with and without replacement bound of Ref. \cite{Diaconis1980} and Eq.~\eqref{eq-bound-on norm1}, we have 
\begin{align}
\label{eq:replace-append-bound}
 \tnorm{\tr_{N-k}\left(\rho_E^{(N)}\right) -\sum_{P\in \mc{S}(E)}\mu_E(P)\,\eta_P^{\otimes k}}\leq \frac{2kd}{N}.
\end{align}
By Eq.~\eqref{eq:pmin-append}, every shell $E$ with $\hat P_E \rho^{(N)} \hat P_E \neq 0$ satisfies
$\min_j P_{\beta_{E/N}}(j) \ge p_{\min}$, so Lemma~\ref{le:lemma1:appendix} applies to each such $E$.
Combining this with Lemma~\ref{le:lemma1:appendix} we complete the proof of Theorem \ref{th:trace-de-finetti}.

\section{de Finetti theorem in terms of relative entropy}
\label{sup:rel-de-finneti}
In the following we present a de Finetti theorem for states invariant under all energy-preserving unitaries using relative entropy. In particular, we first bound the relative entropy between the $k$-qudit marginal of the extremal state $\rho_E^{(N)}\in\mc{S}_H$ and the mixture of energy diagonal states as Lemma~\ref{le:theorem11}. Then we bound the relative entropy between the mixture of energy diagonal states and the mixture of thermal states. Combining this with Lemma \ref{le:theorem11}, we obtain Theorem \ref{th:rel-de-finetti}. We present Lemma~\ref{le:theorem11} and Theorem~\ref{th:rel-de-finetti} below.

\begin{lemma}[Approximation (in relative entropy) of EPU-invariant states as mixture of energy diagonal states]
\label{le:theorem11}
Let $\rho_{E}^{(N)} \in \mc{S}_H$ be the $N$-qudit energy $E$ extremal state. Then, there exists a mixture of products of energy-diagonal states, each with average energy $E/N$, such that any $k$-qudit marginal of $\rho_{E}^{(N)}$ can be approximated for a fixed $k\ll N$ and large $N$ as 
\begin{align}
&\DD`*{\tr_{N-k}\left(\rho_E^{(N)}\right)}{\sum_{P\in \mc{S}(E)}\mu_E(P)\,\eta_P^{\otimes k}}\leq \frac{(d-1)k(k-1)}{2(N-1)(N-k+1)},
\end{align}
where $\mu_E(P)=\frac{\binom{N}{NP}}{g(E)}$ is the fraction of strings with type $P\in \mc{S}(E)$ and $\eta_P=\sum_{i=0}^{d-1}P(i)\ket{E_i}\bra{E_i}$ for all $P\in \mc{S}(E)$.
\end{lemma}

\begin{proof}
\begin{align}
\rho_{E}^{(k)} =\sum_{Q\in \mathcal{P}_k}\left(\sum_{P\in \mc{S}(E)} \mu_E(P)\,\,\hyp\right) \hat{\Pi}_Q.
\end{align}
\begin{align}
\tilde{\eta}_E^{(k)}
=\sum_{Q\in\mathcal{P}_k}\left(\sum_{P\in \mc{S}(E)}\mu_E(P)\,\,\mult{P} \right)\,\hat{\Pi}_Q.
\end{align}
Define two probability distributions $R_{\mathrm{Hyp}}=\sum_{P\in \mc{S}(E)} \mu_E(P)\mathsf{Hyp}_P^{(N,k)}$ and $R_{\mathrm{Mult}}=\sum_{P\in \mc{S}(E)} \mu_E(P)\mathsf{Mult}_P^{(N,k)}$ that take values $\sum_{P\in \mc{S}(E)} \mu_E(P)\mathsf{Hyp}_P^{(N,k)}(Q)$ and $\sum_{P\in \mc{S}(E)} \mu_E(P)\mathsf{Mult}_P^{(N,k)}(Q)$, respectively for $Q\in \mc{P}_k$. Then
\begin{align*}
   \DD`*{\rho_{E}^{(k)}}{\tilde{\eta}_E^{(k)}}
   &=\DD`*{R_{\mathrm{Hyp}}}{R_{\mathrm{Mult}}}\\
   \\
    &=\DD`*{\sum_{P\in \mc{S}(E)} \mu_E(P)\mathsf{Hyp}_P^{(N,k)}}{\sum_{P\in \mc{S}(E)} \mu_E(P)\mathsf{Mult}_P^{(N,k)}}\\
    &\leq \sum_{P\in \mc{S}(E)} \mu_E(P)\DD`*{\mathsf{Hyp}_P^{(N,k)}}{\mathsf{Mult}_P^{(N,k)}}\nonumber\\
    &\leq \frac{(d-1)k(k-1)}{2(N-1)(N-k+1)},
\end{align*}
where in the last line we used the bound on the relative entropy between the sampling from with and without replacement. In particular, Ref. \cite{Johnson2024} proves that $\DD`*{\mathsf{Hyp}_P^{(N,k)}}{\mathsf{Mult}_P^{(N,k)}}\leq \frac{(d-1)k(k-1)}{2(N-1)(N-k+1)}$ (see also Refs. \cite{Stam1978, Borderi2022, Gavalakis2024}). This completes the proof of the lemma.
\end{proof}

\begin{theorem}[de Finetti theorem for EPU invariant states]
\label{th:rel-de-finetti}
Let $\rho^{(N)}=\sum_{E\in \mathsf{Spec}(H)} c_E \rho^{(N)}_E$ be an $N$-qudit state which is invariant under all energy preserving unitaries. Assume that $p_{\min}\left(\rho^{(N)}\right) > 0$. Then, for any $k$-qudit marginal of such a state, there exists a mixture of $k$-qudit thermal states $\{\tau_{\beta}^{\otimes k}\}_{\beta}$ such that for a fixed $k\ll \sqrt{N}$ and large $N$
\begin{align}
\DD`*{\tr_{N-k}\left(\rho^{(N)}\right)}{\int\mu(\dbeta)\,\tau_\beta^{\otimes k}} &\leq \frac{k (d-2)}{2N} + O\left(N^{-3/2+c\delta}\right),
\end{align}
where $c>0$ is some $d$-dependent constant,  $\delta\in\left(0,1/2c\right)$, $\tau_\beta$ is the thermal state of single qudit at inverse temperature $\beta$ and $\mu(\dbeta)$ is a valid probability measure on the set $\{\beta\}$ of inverse temperatures. 
\end{theorem}

\begin{proof}
Let $\rho^{(N)}$ be an arbitrary state on $N$ qudits that commutes with all energy preserving unitaries. Then we can write it as
\begin{align}
    \rho^{(N)} = \sum_{E\in \mathsf{Spec}(H)} c_E \rho^{(N)}_E.
\end{align}
Partial tracing $(N-k)$ qudits from $\rho^{(N)}$ results in
\begin{align}
   \tr_{N-k}\left( \rho^{(N)}\right) = \sum_{E\in \mathsf{Spec}(H)} c_E  \tr_{N-k}\left( \rho^{(N)}_E\right).
\end{align}
Let $\mu(\dbeta):=\sum_{E\in\spec(H)}c_E\,\delta (\beta-\beta_{E/N})\dbeta$, where $\beta_{E/N}$ is the inverse temperature obtained by solving $\sum_{i=0}^{d-1}E_i\e^{-\beta_{E/N}E_i}=(E/N)\sum_{i=0}^{d-1}\e^{-\beta_{E/N}E_i}$. 
Then
\begin{align}
   \DD`*{\tr_{N-k}\left( \rho^{(N)}\right)}{\int_{\beta}\mu(\dbeta)\,\tau_{\beta}^{\otimes k}} &= \DD`*{\sum_{E\in \mathsf{Spec}(H)} c_E  \tr_{N-k}\left( \rho^{(N)}_E\right)}{\sum_{E\in \mathsf{Spec}(H)} c_E \,\tau{\scriptstyle(E/N)}^{\otimes k}}\nonumber\\
   &\leq \sum_{E\in \mathsf{Spec}(H)} c_E \, \DD`*{\tr_{N-k}\left( \rho^{(N)}_E\right)}{\tau{\scriptstyle(E/N)}^{\otimes k}},
\end{align}
where $\tau{\scriptstyle(E/N)}=\tau_{\beta_{E/N}}$. Now upper bounding $\DD`*{\tr_{N-k}\left( \rho^{(N)}_E\right)}{\tau{\scriptstyle(E/N)}^{\otimes k}}$ will suffice to prove the theorem. Let $\rho_{i_1,\cdots,i_k}=\bra{i_1,\cdots,i_k}\tr_{N-k}\left( \rho^{(N)}_E\right)\ket{i_1,\cdots,i_k}$, then
\begin{align*}
 &\DD`*{\tr_{N-k}\left( \rho^{(N)}_E\right)}{\tau{\scriptstyle(E/N)}^{\otimes k}}\\
 &= \DD`*{\tr_{N-k}\left( \rho^{(N)}_E\right)}{\sum_{P\in \mc{S}(E)} \mu_E(P)\,\eta_P^{\otimes k}}+ \DD`*{\sum_{P\in \mc{S}(E)} \mu_E(P)\,\eta_P^{\otimes k}}{\tau{\scriptstyle(E/N)}^{\otimes k}}+R_E,   
\end{align*}
where 
\begin{align}
\label{eq:R_E-term}
R_E := \sum_{i_1,\cdots,i_k=0}^{d-1} \left(\rho_{i_1,\cdots, i_k} -\sum_{P\in\mc{S}(E)}\mu_E(P)\,\prod_{r=1}^kP(i_r)\right)\,
        \left[ \ln \sum_{P\in\mc{S}(E)}\mu_E(P)\, \prod_{r=1}^kP(i_r) - \ln \prod_{r=1}^kP_\beta(i_r)\right].
\end{align}
From Lemma \ref{le:theorem11}, we already have
$\DD`*{\tr_{N-k}\left( \rho^{(N)}_E\right)}{\sum_{P\in \mc{S}(E)} \mu_E(P)\,\eta_P^{\otimes k}}\le \frac{(d-1)k(k-1)}{2(N-1)(N-k+1)}$. Further, from the joint convexity of the relative entropy, we have $\DD`*{\sum_{P\in \mc{S}(E)} \mu_E(P)\,\eta_P^{\otimes k}}{\tau{\scriptstyle(E/N)}^{\otimes k}}\leq k\sum_{P\in \mc{S}(E)}\mu_E(P) \DD`*{P} {P_{\beta}}$. Note that for all $i=0,\cdots, d-1$, $P_\beta(i)\geq p_{\min}>0$. Consider the following expansion valid for $\left|P(i)-P_\beta(i)\right|\leq N^{-1/2+\delta}$, i.e. for $P \in \mathcal{S}(E) \cap \mathcal{B}_\delta$.
\begin{align*}
   \DD`*{P} {P_{\beta}} 
   &= \sum_{i=0}^{d-1}P_\beta(i)\left(1+ \frac{P(i)-P_\beta(i)}{P_\beta(i)}\right)\ln \left(1+ \frac{P(i)-P_\beta(i)}{P_\beta(i)}\right)\\
   &= \frac{1}{2}\sum_{i=0}^{d-1}\frac{(P(i)-P_\beta(i))^2}{P_\beta(i)} + O\left(\norm{P-P_\beta}_2^3\right).
\end{align*}
Outside $\mathcal{B}_\delta$ the expansion is not needed: by Eq.~\eqref{eq:dropping-outside-eq},
$\mu_E\bigl(\mathcal{S}(E)\setminus\mathcal{B}_\delta\bigr) = O\bigl(e^{-N^{2\delta}/2}\bigr)$,
while $D(P\|P_\beta) \le -\sum_i P(i)\ln P_\beta(i) \le \ln(1/p_{\min})$ uniformly on
$\Delta(E)$, so the complement contributes $O\bigl(e^{-N^{2\delta}/2}\bigr)$. Summing over
$P \in \mathcal{S}(E)$, we obtain
\begin{align*}
   \sum_{P\in\mc{S}(E)}\mu_E(P)\,\DD`*{P} {P_{\beta}} &= \frac{\sum_{P\in\mc{S}(E)}\e^{NH(P)} \prod_j P(j)^{-1/2}\DD`*{P} {P_{\beta}}}{\sum_{Q\in \mathcal{S}(E)}\e^{NH(Q)} \prod_j Q(j)^{-1/2} }\left(1+O(1/N)\right).
\end{align*}
Then the leading order contribution to the above expression comes from
\begin{align*}
   \sum_{P\in\mc{S}(E)}\mu_E(P)\,\DD`*{P} {P_{\beta}} &= \frac{1}{2}\sum_{i=0}^{d-1}\frac{1}{P_\beta(i)}\frac{\sum_{P\in\mc{S}(E) \cap \mathcal{B}_\delta}\e^{-\frac{N}{2}(P-P_\beta)^TK(P-P_\beta)} (P(i)-P_\beta(i))^2  }{\sum_{P\in \mathcal{S}(E) \cap \mathcal{B}_\delta}\e^{-\frac{N}{2}(P-P_\beta)^TK(P-P_\beta)}} + O\left(N^{-3/2+c\delta}\right).
\end{align*}
Converting this expression into $\pmb{k}\in\mathbb Z^{d-2}-\pmb{k}_\beta$ coordinates by noting $P=\pmb{u}/N+P_\beta = B\pmb{w}/N+P_\beta = BJ\pmb{k}/N+P_\beta$ and $\wt{K}=J^TB^TKBJ$, we obtain (see Appendix \ref{subsec:main-results})
\begin{align*}
\sum_{P\in\mc{S}(E)}\mu_E(P)\,\DD`*{P} {P_{\beta}}&=\frac{1}{2N^2}\sum_{i=0}^{d-1}P_\beta(i)^{-1}\frac{\sum_{\pmb{k}\in\mathbb Z^{d-2}-\pmb{k}_\beta} \e^{-\frac{1}{2N}\pmb{k}^T\wt{K} \pmb{k}} (b_i^T{\pmb k})^2}{\sum_{\pmb{k}\in\mathbb Z^{d-2}-\pmb{k}_\beta} \e^{-\frac{1}{2N}\pmb{k}^T\wt{K} \pmb{k}}  } + O\left(N^{-3/2+c\delta}\right)\\
&=\frac{1}{2N}\sum_{i=0}^{d-1}P_\beta(i)^{-1}\,b_i^T\wt{K}^{-1}b_i+ O\left(N^{-3/2+c\delta}\right)\\
  &= \frac{1}{2N}\tr\left(\mathsf{Proj}_KK^{-1} K\right)+ O\left(N^{-3/2+c\delta}\right)\\
  &= \frac{(d-2)}{2N}+ O\left(N^{-3/2+c\delta}\right),
\end{align*}
where we used Eqs. \eqref{eq:zeroth} and \eqref{eq:second} and the fact  that $b_i^T\wt{K}^{-1}b_i = \mathsf{Mat}_{ii}=\left(\mathsf{Proj}_KK^{-1}\right)_{ii}$. Then, we get
\begin{align}
\DD`*{\sum_{P\in \mc{S}(E)} \mu_E(P)\,\eta_P^{\otimes k}}{\tau{\scriptstyle(E/N)}^{\otimes k}} \leq \frac{k(d-2)}{2N}+ O\left(N^{-3/2+c\delta}\right).
\end{align}

Now we bound the remainder term $R_E$ (Eq.~\eqref{eq:R_E-term}) as follows. Consider the following expression 
\begin{align*}
&\ln \sum_{P\in\mc{S}(E)}\mu_E(P)\,  P(i_1)\cdots P(i_k)-\ln  \left(\prod_{r=1}^kP_\beta(i_r)\right) \\
&= \sum_{r=1}^k \sum_{P\in\mc{S}(E)}\mu_E(P)\frac{P(i_r)-P_\beta(i_r)}{P_\beta(i_r)}
+\sum_{1\leq r<s\leq k}\sum_{P\in\mc{S}(E)}\mu_E(P) \frac{\left(P(i_r)-P_\beta(i_r)\right) \left(P(i_s)-P_\beta(i_s)\right)}{P_\beta(i_r)P_\beta(i_s)}\\
&-\frac{1}{2}\left(\sum_{r=1}^k\sum_{P\in\mc{S}(E)}\mu_E(P)\frac{P(i_r)-P_\beta(i_r)}{P_\beta(i_r)}\right)^2 +\text{higher order terms}.
\end{align*}
We bound the above expression term by term to get the leading order corrections. To the leading order, using the chain of equalities in Eq. \eqref{eq:chain-for-linear-terms}, we get
\begin{align*}
\sum_{P\in\mc{S}(E)}\mu_E(P)\frac{P(i_r)-P_\beta(i_r)}{P_\beta(i_r)}= -\frac{1}{2N}\sum_j\frac{\left[1-P_\beta(j)^{-1}\,\mathsf{Mat}_{jj}\right]\mathsf{Mat}_{i_rj}}{P_\beta(j) P_\beta(i_r) }
 +O\left(N^{-3/2+c\delta}\right)
\end{align*}
Similarly, using chain the of equalities in Eq. \eqref{eq:chain-of-inequalities-quadratic}, to the leading order, we obtain
\begin{align*}
\sum_{P\in\mc{S}(E)}\mu_E(P)\frac{(P(i_r)-P_\beta(i_r))(P(i_s)-P_\beta(i_s))}{P_\beta(i_r)P_\beta(i_s)}
&=\frac{1}{N} \frac{\mathsf{Mat}_{i_ri_s} }{P_\beta(i_r)P_\beta(i_s)} +O\left(N^{-3/2+c\delta}\right).
\end{align*}
Thus,
\begin{align}
\label{eq:reminder-lead}
&\left|\ln \sum_{P\in\mc{S}(E)}\mu_E(P)\,  P(i_1)\cdots P(i_k)-\ln  \left(\prod_{r=1}^kP_\beta(i_r)\right)\right| \nonumber\\
&\leq  \frac{1}{2N}\sum_{r=1}^k\sum_j \left|\frac{\left[1-P_\beta(j)^{-1}\,\mathsf{Mat}_{jj}\right]\mathsf{Mat}_{i_rj}}{P_\beta(j) P_\beta(i_r) }\right|+ \frac{1}{N} \sum_{1\leq r<s\leq k}  \left|\frac{\mathsf{Mat}_{i_ri_s} }{P_\beta(i_r)P_\beta(i_s)}\right|+O\left(N^{-3/2+c\delta}\right)\nonumber\\
&\leq  \frac{1}{2N}\sum_{r=1}^k\sum_j \left|\frac{\mathsf{Mat}_{i_rj}}{P_\beta(j) P_\beta(i_r) }\right|+ \frac{1}{N} \sum_{1\leq r<s\leq k}  \left|\frac{\mathsf{Mat}_{i_ri_s} }{P_\beta(i_r)P_\beta(i_s)}\right|+O\left(N^{-3/2+c\delta}\right),
\end{align}
where in the last line we used the fact that  $1-P_\beta(j)^{-1}\,\mathsf{Mat}_{jj} = q_j\leq 1$ (see the proof of Lemma \ref{lemma:append-restriction-to-T}). Now we upper bound the term  $|\mathsf{Mat}_{i_ri_s}|$. Let $e_i^T=(0\cdots,1,\cdots,0)$. Then applying Cauchy-Schwarz inequality, we have $|\mathsf{Mat}_{ij}| \leq \sqrt{e_i^T \mathsf{Mat}\, e_i}\sqrt{e_j^T \mathsf{Mat}\, e_j}\leq \sqrt{P_\beta(i)P_\beta(j)}$, where we have used the fact that $\mathsf{Mat}\leq K^{-1}$. Thus, $\frac{1}{P_{\beta}(i_r)P_\beta(i_s)} \mathsf{Mat}_{i_ri_s} \leq 1/\sqrt{P_\beta(i_r)P_\beta(i_s)}\leq 1/p_{\min}$. Then from Eq. \eqref{eq:reminder-lead}, we have
\begin{align*}
\left|\ln \sum_{P\in\mc{S}(E)}\mu_E(P)\,  P(i_1)\cdots P(i_k)-\ln  \left(\prod_{r=1}^kP_\beta(i_r)\right)\right|
&\leq\frac{k(d+k-1)}{2Np_{\min}}.
\end{align*}
Thus,
\begin{align*}
|R_E| &\leq \frac{k(d+k-1)}{2Np_{\min}}\sum_{i_1,\cdots,i_k=0}^{d-1} \left|\left(\rho_{i_1,\cdots, i_k} -\sum_{P\in\mc{S}(E)}\mu_E(P)\,P(i_1)\cdots P(i_k)\right)\right|\\
&\leq \frac{k(d+k-1)}{2Np_{\min}}\tnorm{\tr_{N-k}\left( \rho^{(N)}_E\right)-\sum_{P\in \mc{S}(E)} \mu_E(P)\,\eta_P^{\otimes k}}\\
&\leq \frac{k^2d(d+k-1)}{N^2p_{\min}},
\end{align*}
where we have used $\tnorm{\tr_{N-k}\left( \rho^{(N)}_E\right)-\sum_{P\in \mc{S}(E)} \mu_E(P)\,\eta_P^{\otimes k}}\leq 2kd/N$ from Eq. \eqref{eq:replace-append-bound}. Now
\begin{align}
\DD`*{\tr_{N-k}\left(\rho^{(N)}\right)}{\int\mu(\dbeta)\,\tau_\beta^{\otimes k}}&\leq \frac{k (d-2)}{2N}+ \frac{k^2d(d+k-1)}{N^2p_{\min}}+ \frac{(d-1)k(k-1)}{2(N-1)(N-k+1)}+ O\left(N^{-3/2+c\delta}\right).
\end{align}
The remainder term $|R_E|\leq \frac{k^2d(d+k-1)}{N^2p_{\min}}$ and the term $\DD`*{\tr_{N-k}\left( \rho^{(N)}_E\right)}{\sum_{P\in \mc{S}(E)} \mu_E(P)\,\eta_P^{\otimes k}}\leq \frac{(d-1)k(k-1)}{2(N-1)(N-k+1)}$ can be absorbed into the correction terms of order  $O\left(N^{-3/2+c\delta}\right)$. Note that the worst case scaling for $k$ and $N$ is of the form $(k^2/N)^{i/2}$ with $i>0$ (see also Remark \ref{rem:worst-scaling}). Combining everything, we have
\begin{align}
\DD`*{\tr_{N-k}\left(\rho^{(N)}\right)}{\int\mu(\dbeta)\,\tau_\beta^{\otimes k}}&\leq \frac{k (d-2)}{2N}+ O\left(N^{-3/2+c\delta}\right).
\end{align}
This concludes the proof of the theorem.
\end{proof}

\section{Quantifying Symmetry Breaking and Robust Thermalization}
\label{eq:Lindblad}
\subsection{Approximate symmetry}
 Here we show how the asymmetry of a state, quantifying the departure of the state from being invariant under all energy-preserving unitaries, controls the closeness of a state to a mixture of thermal states via a de Finetti approximation. We first comment on the time translation symmetry which is strictly weaker than the EPU invariance.

\

\medskip
\noindent
{\bf EPU invariance versus time-translation invariance.} A quantum state $\rho$ is said to be invariant under time translations if it commutes with the system Hamiltonian, i.e., $[\rho, H] = 0$. This corresponds to invariance under the one-parameter unitary group $\{ e^{-iHt} \}_{t \in \mathbb{R}}$, representing time evolution. In contrast, EPU invariance requires $\rho$  to be invariant under \emph{all} unitaries $U \in \mathcal{U}_H$ satisfying $[U, H] = 0$, i.e., under the full commutant of $H$. While time-translation invariance enforces block-diagonality in the energy eigenbasis, EPU invariance is strictly stronger. It demands $\rho$ to be invariant under arbitrary unitaries acting within each energy eigenspace, i.e., that $\rho$ is maximally mixed on each energy shell. This distinction is operationally significant in quantum thermodynamics and the resource theory of asymmetry. In the resource-theoretic framework, asymmetry relative to time-translation symmetry is identified with coherence between energy levels~\cite{Marvian2014, Lostaglio2015, Gour2008}, and EPU-invariant states represent a strict subset of time-translation-invariant states that lack not only coherence, but also any internal structure within energy eigenspaces that could be exploited by symmetric operations. In the present work, this stronger symmetry plays a central role in determining the thermal structure of subsystems via a de Finetti-type approximation.

\medskip
\noindent
\textbf{EPU Twirling.}  Let us define EPU twirling operation $\mathcal{T}_{\mathrm{EPU}}$ as 
\begin{align}
\mathcal{T}_{\mathrm{EPU}}(\rho)
:= \int_{\mc{U}_H} \mathrm{d}U \, U\rho U^\dagger,
\end{align}
where $\mathrm{d}U$ denotes the normalized Haar measure on $\mc{U}_H$ and $\rho\in \mc{D}_N$. Since the group factorizes across energy eigenspaces, for any state $\rho$ we obtain
\begin{align}
\mathcal{T}_{\mathrm{EPU}}(\rho)
= \bigoplus_E \mathrm{Tr}(\hat P_E \rho)\, \frac{\hat P_E}{g(E)},
\end{align}
where $g(E) = \mathrm{tr}(\hat P_E)$. This state is invariant under all $U\in\mc{U}_H$, and by Theorem~\ref{th:trace-de-finetti}, it admits an approximate decomposition as a convex mixture of thermal states as long as $p_{\min}(\mathcal{T}_{\mathrm{EPU}}(\rho))>0$.

\medskip
\noindent
{\bf Symmetry Breaking Measure.} Let us define the symmetry breaking distance of \(\rho\) from exact EPU invariance via quantum relative entropy as
\begin{align}
\Delta_{\mathrm{asym}}(\rho) := D\left(\rho ||\mathcal{T}_{\mathrm{EPU}}(\rho)\right).
\end{align}
The measure $\Delta_{\mathrm{asym}}(\rho)$ quantifies the degree to which $\rho$ deviates from EPU invariance. We remark that $\Delta_{\mathrm{asym}}$ is actually an asymmetry monotone (under EPU covariant operations) \cite{Gour2008}. By construction, $\mathcal{T}_{\mathrm{EPU}}(\rho)$ is EPU-invariant and thus admits the de Finetti-type approximation. We now show that (in Theorem \ref{th:theorem-next} below) the original state $\rho$, if close to EPU-invariant, inherits an approximate de Finetti representation with an error explicitly controlled  by $\Delta_{\mathrm{asym}}\left(\rho\right)$.

\begin{theorem}[Robust thermalization with symmetry breaking]
\label{th:theorem-next}
Let $\rho^{(N)}$ be a state of $N$-qudit system with Hamiltonian $H = \sum_{j=1}^N h^{(j)}$ satisfying  $p_{\min}\left(\rho^{(N)}\right)>0$. Then the $k$-qudit marginal of $\rho^{(N)}$ can be approximated for a fixed $k\ll \sqrt{N}$ and large $N$ as
\begin{align}
    \left\| \tr_{(N-k)}\left(\rho^{(N)}\right) - \int\mu(\dbeta)\,\tau_{\beta}^{\otimes k} \right\|_1 \leq \sqrt{2\Delta_{\mathrm{asym}}\left(\rho^{(N)}\right)}+\frac{k(8\kappa+d(k+3)-2k+6)}{2N}+  O\left(N^{-3/2+c\delta}\right),
\end{align}
where $c>0$ is some $d$-dependent constant,  $\delta\in\left(0,1/2c\right)$ and $\mu$ is a probability measure on the set of inverse temperatures $\beta$.
\end{theorem}
\begin{proof}
The proof of the theorem follows by using triangle inequality, Pinsker's inequality \cite{CoverThomas2006} and Theorem~\ref{th:trace-de-finetti}.
\begin{align*}
&\tnorm{\tr_{(N-k)}\left(\rho^{(N)}\right) - \int\mu(\dbeta)\,\tau_{\beta}^{\otimes k}} \\
&\leq\tnorm{\tr_{(N-k)}\left(\rho^{(N)}\right)- \tr_{(N-k)}\left(\mc{T}_{\mathrm{EPU}}\left(\rho^{(N)}\right)\right)}+\tnorm{\tr_{(N-k)}\left(\mc{T}_{\mathrm{EPU}}\left(\rho^{(N)}\right)\right)- \int\mu(\dbeta)\,\tau_{\beta}^{\otimes k}}\\
&\leq\tnorm{\rho^{(N)}- \mc{T}_{\mathrm{EPU}}\left(\rho^{(N)}\right)}+\frac{k(8\kappa+d(k+3)-2k+6)}{2N}+  O\left(N^{-3/2+c\delta}\right)\\
&\leq\sqrt{2\Delta_{\mathrm{asym}}\left(\rho^{(N)}\right)}+ \frac{k(8\kappa+d(k+3)-2k+6)}{2N}+  O\left(N^{-3/2+c\delta}\right).
\end{align*}
\end{proof}

\subsection{EPU-invariance from Lindblad dynamics}

We consider a physically motivated class of open-system dynamics that converges to an \emph{exactly} EPU-invariant state, thereby enabling a direct application of our de Finetti theorem. Let a system of $N$ qudits have Hamiltonian $H=\sum_{j=1}^N h^{(j)}$, where $h^{(j)} = \sum_{x=0}^{d-1} E_x \ket{x}\bra{x}$. $\{\ket{x}\}_{x=0}^{d-1}$ is the local energy eigenbasis and the $E_x$ are distinct. The global energy eigenbasis is therefore $\{ \ket{\Bx} := \ket{x_1,\ldots,x_N} \}_{\Bx\in \{0,\ldots,d-1\}^N}$.
For each total energy $E$, define the projector $\hat P_E := \sum_{\Bx: \,\sum_i E_{x_i}=E} \ket{\Bx}\bra{\Bx}$ with $\mathrm{tr}(\hat P_E)=g(E)$. Let the initial state of the system be given by $\rho_0$.

\medskip
\noindent
{\bf The physical model.} The $N$ qudits interact with an environment through brief, sudden collisions occurring at random instants. A collision is of one of $n$ possible types; a collision of type $j\in\{1,\dots,n\}$ applies a unitary $V$ drawn from a fixed distribution $\mathsf{dist}_j$, whose average channel is $\Cc_j(\sigma):=\E_{\mathsf{dist}_j}[\,V\sigma V^\dagger\,]$. Collisions of type $j$ arrive as an independent Poisson process of rate $r_j$ and in disjoint time intervals the numbers of collisions are independent. The number of collisions in an interval of length
$t$ is Poisson-distributed with mean $r_j t$. We work in the interaction picture with respect to
$H$, so that between collisions the state does not change. A single noise realization over $[0,t]$ is
$\xi=(t_1<\dots<t_m;\,V_1,\dots,V_m)$. Along it the state is piecewise constant, jumping at each collision, $\rho_\xi(t)=V_m\cdots V_1\,\rho_0\,V_1^\dagger\cdots V_m^\dagger$. Each $\rho_\xi(t)$ is a density matrix, but it is random and discontinuous and has no smooth equation of motion. The object that does is the noise average $\bar\rho(t):=\E_\xi[\rho_\xi(t)]$.

\bigskip
\noindent
{\bf The $n$-process master equation.} Note two elementary facts about Poisson processes
\cite{Kingman1993}: (1) the union of the $n$ independent processes is a single Poisson process of total rate
$R=\sum_{j=1}^n r_j$; and (2) each collision of the combined process is, independently, of type $j$ with
probability $p_j=r_j/R$. Define the combined single-collision channel $\Cc$ as
\begin{align}
\label{eq:Ccomb}
\Cc:=\sum_{j=1}^n p_j\,\Cc_j=\frac1R\sum_{j=1}^n r_j\,\Cc_j.
\end{align}
Conditioning on the total number $m$ of collisions in $[0,t]$ and averaging over the drawn unitaries and over types, we get $\E_{\text{types},\{\mathsf{dist}_j\}}\!\left[V_m\cdots V_1\,\rho_0\,V_1^\dagger\cdots V_m^\dagger\right]
=\Cc^{\,m}(\rho_0)$, independently of the collision times. Only the
number $m$ of collisions survives; it is Poisson-distributed with mean $Rt$. Hence, for all $t\ge0$,
\begin{align}
\label{eq:exact}
\bar\rho(t)=\sum_{m=0}^{\infty}e^{-Rt}\frac{(Rt)^m}{m!}\,\Cc^{\,m}(\rho_0)
=e^{t\Lc}(\rho_0),
\end{align}
where $\Lc=\sum_{j=1}^{n}r_j\,(\Cc_j-\id)$. Eq. \eqref{eq:exact} exhibits $e^{t\Lc}$ as a Poisson-weighted convex mixture of the channels $\Cc^{\,m}$, so it is completely positive and trace preserving at every $t$, and $\Lc$ is a legitimate Lindblad generator.

Let $\Qe=\mathbb{I}-\Pe$, then for any state $\rho$, we have
\begin{align}
\label{eq:blocks}
\rho=\Pe\rho\Pe+\Pe\rho\Qe+\Qe\rho\Pe+\Qe\rho\Qe.
\end{align}
Consider the following two types of collisions. Both are block diagonal in the energy shells and hence are energy-preserving. 

\medskip
\noindent
{\it Type 1: randomization (Haar) collision on shell $E$, rate $\gamma_E$.}
The unitary is $V=V_E\oplus\mathbb{I}$, with $V_E$ Haar-random on shell $E$ and the identity on its
complement. Conjugating Eq. \eqref{eq:blocks} and averaging over $V_E$, the three kinds of block behave
differently: the Haar twirl depolarizes $\Pe\rho\Pe$ to the microcanonical state of the shell; the cross
blocks carry a single factor of $V_E$ and are annihilated by $\int dV_E\, V_E=0$; and $\Qe\rho\Qe$
is untouched. Hence the average channel is
\begin{align}
\label{eq:Vchannel}
\Vc_E(\rho)=\frac{\tr(\Pe\rho \Pe)}{\gE}\,\Pe+\Qe\rho\,\Qe.
\end{align}

\medskip
\noindent
{\it Type 2: phase collision on shell $E$, rate $\nu_E$.}
It applies a random overall phase to the shell $E$ and leaves everything else untouched, i.e., $U^{(E)}_\theta=e^{i\theta}\Pe+\Qe$, where $\theta$ is uniform on $[0,2\pi)$. Expanding
$U^{(E)}_\theta\rho\,U^{(E)\dagger}_\theta=\Pe\rho\Pe+\Qe\rho\Qe+e^{i\theta}\Pe\rho\Qe+e^{-i\theta}\Qe\rho\Pe$
and then averaging over $\theta$ gives
\begin{align}
\label{eq:Wchannel}
\Wc_E(\rho)=\int_0^{2\pi}\!\frac{d\theta}{2\pi}\,U^{(E)}_\theta\rho\,U^{(E)\dagger}_\theta
=\Pe\rho\Pe+\Qe\rho\Qe.
\end{align}
Relabeling $j\leftrightarrow(E,\text{kind})$, the full Lindbladian, for the two types of collisions, reads as
\begin{align}
\label{eq:relabel}
\Lc (\rho) &=\sum_{E}\gamma_E\left(\Vc_E(\rho)-\rho\right) + \sum_E\nu_E\left(\Wc_E(\rho)-\rho\right)\nonumber\\
&=\sum_E\gamma_E\left[\tfrac{\tr(\Pe\rho \Pe)}{\gE}\Pe-\Pe\rho\Pe\right]-\sum_E\left(\gamma_E+\nu_E\right)\left(\Pe\rho\Qe+\Qe\rho\Pe\right)\nonumber\\
&=\sum_E\gamma_E\left[\tfrac{\tr(\Pe\rho \Pe)}{\gE}\Pe-\Pe\rho\Pe\right]+\sum_E2\left(\gamma_E+\nu_E\right)\left[\Pe\rho\Pe-\frac{1}{2}\left\{\Pe,\rho\right\}\right].
\end{align}
Choosing $2(\gamma_E+\nu_E)=\lambda_E$ for every shell, we get $\Lc=\mathcal{L}_{\mathrm{block}} + \mathcal{L}_{\mathrm{deph}}$, where 
\begin{align}
&\mathcal{L}_{\mathrm{block}}(\rho)=\sum_E\gamma_E\!\left[\tfrac{\tr\left(\Pe\rho\Pe\right)}{\gE}\Pe-\Pe\rho\Pe\right],\\
&\mathcal{L}_{\mathrm{deph}}(\rho)=\sum_E\lambda_E\!\left[\Pe\rho\Pe-\tfrac12\{\Pe,\rho\}\right].
\end{align}
When $\lambda_E\ge2\gamma_E$, the phase rates $\nu_E=\tfrac12\lambda_E-\gamma_E$ are nonnegative, so this construction is realizable.

\begin{remark}[Interaction vs.\ Schr\"odinger picture]
\label{rem:picture}
Since the collisions are energy-preserving, the free evolutions collect into
$U_0(t)=e^{-iHt}$ and factor out of every realization, giving $\bar\rho^{\rm
Sch}(t)=U_0(t)\,\bar\rho(t)\,U_0(t)^\dagger$. Differentiating and using
$U_0\,\mathcal L(\rho)\,U_0^\dagger=\mathcal L(U_0\rho U_0^\dagger)$ (each term
of $\mathcal L$ is built from the $\hat P_E$, which commute with $H$) gives the
Schr\"odinger-picture master equation $\dot{\bar\rho}^{\rm Sch}=-i[H,\bar\rho^{\rm
Sch}]+\mathcal L(\bar\rho^{\rm Sch})$. The fixed point of $\mathcal L$ commutes
with $H$ and is left invariant by the free evolution, so the Hamiltonian does
not affect the long-time behavior and one may work in the interaction picture
without loss of generality.
\end{remark}

The term $\mathcal{L}_{\mathrm{block}}$ drives each block $\hat P_E \rho \hat P_E$ toward the maximally mixed state $\hat P_E/g(E)$, while $\mathcal{L}_{\mathrm{deph}}$ damps all coherences $\hat P_E \rho \hat P_{E'}$ for $E\neq E'$. Both diagonal and off-diagonal subspaces are invariant under $\mathcal{L}$, and one verifies that $[\mathcal{L}_{\mathrm{block}},\,\mathcal{L}_{\mathrm{deph}}]=0$, 
so the semigroup factorizes as
\begin{align}
e^{t\mathcal{L}} = e^{t\mathcal{L}_{\mathrm{block}}}\,e^{t\mathcal{L}_{\mathrm{deph}}}
= e^{t\mathcal{L}_{\mathrm{deph}}}\,e^{t\mathcal{L}_{\mathrm{block}}}.
\end{align}

\medskip
\noindent\textbf{Evolution of diagonal blocks.}
Under Lindbladian dynamics, we have $\rho(t) := e^{t\mathcal{L}}(\rho_0)$ or $\frac{d}{dt}\rho(t)=\mc{L}\rho(t)$. Let 
\begin{align} 
    \rho(t)=\sum_{E}\hat P_E\rho(t)\hat P_E +\sum_{E\ne E'} \hat P_E\rho(t)\hat P_{E'}.
\end{align}
Define $D_E(t):=\hat P_E\rho(t)\hat P_E$ and $O_{EE'}(t):=\hat P_E\rho(t)\hat P_{E'}$ for $E\ne E'$. Then from~\eqref{eq:Lblock-rig} and ~\eqref{eq:Ldeph-rig},
\begin{align*}
    \frac{d}{dt} D_{E}(t)&=\hat P_{E}\left[\mc{L}_{\mathrm{block}}(\rho(t)) + \mc{L}_{\mathrm{deph}}(\rho(t))\right]\hat P_{E}\\
    &=\hat P_{E}\sum_{E''} \gamma_{E''}
\left[
\frac{\mathrm{tr}(\hat P_{E''} \rho(t) \hat P_{E''})}{g(E'')}\, \hat P_{E''} - \hat P_{E''} \rho(t) \hat P_{E''}
\right] \hat P_{E}\\
&\quad+\hat P_{E} \sum_{E''} \lambda_{E''} 
\left[
\hat P_{E''} \rho(t) \hat P_{E''} - \tfrac{1}{2}\{\hat P_{E''},\rho(t)\}
\right]\hat P_{E}\\
&=\gamma_{E}
\left[
\mathrm{tr}(D_{E}(t))\, \frac{\hat P_{E}}{g(E)} - D_{E}(t)
\right].
\end{align*}
Note that $\frac{d}{dt}\tr\left(D_{E}(t)\right)=0$, thus $\tr\left(D_{E}(t)\right)=\tr\left(D_{E}(0)\right)$. Then taking $\Delta_{E}(t)=\left[
\mathrm{tr}(D_{E}(0))\, \frac{\hat P_{E}}{g(E)} - D_{E}(t)
\right]$, we have
\begin{align*}
    \frac{d}{dt} \Delta_{E}(t)&=-\gamma_{E}
\Delta_{E}(t).
\end{align*}
Thus, $\Delta_{E}(t)= \Delta_{E}(0)e^{-t\gamma_{E}}$ or
\begin{align}
\label{eq:BE-solution}
    D_{E}(t)&=\mathrm{tr}(D_{E}(0))\, \frac{\hat P_{E}}{g(E)} - \left(\mathrm{tr}(D_{E}(0))\, \frac{\hat P_{E}}{g(E)} - D_{E}(0)\right)e^{-t\gamma_{E}}\nonumber\\
    &=(1-e^{-t\gamma_{E}})\mathrm{tr}(D_{E}(0))\, \frac{\hat P_{E}}{g(E)} + e^{-t\gamma_{E}} D_{E}(0).
\end{align}

\medskip
\noindent\textbf{Evolution of off-diagonal blocks.}
We have
\begin{align*}
    \frac{d}{dt} O_{E,E'}(t) &= \hat P_E\left(\mc{L}_{\mathrm{block}}(\rho(t))+\mc{L}_{\mathrm{deph}}(\rho(t))\right)\hat P_{E'}\\
    &=
-\tfrac{1}{2}(\lambda_E + \lambda_{E'})\, O_{E,E'}(t)
\end{align*}
Then
\begin{align}
O_{E,E'}(t)
=e^{-\frac{1}{2}(\lambda_E+\lambda_{E'})t}\,O_{E,E'}(0),
\qquad E\neq E'.
\label{eq:CE-solution}
\end{align}

Combining~\eqref{eq:BE-solution} and~\eqref{eq:CE-solution}, the exact state at time $t$ is
\begin{align}
\rho(t)
&=
\sum_{E}\left[(1-\varepsilon_E(t))\,\hat P_E\rho_0\hat P_E+\varepsilon_E(t)\,\mathrm{tr}(\hat P_E\rho_0)\,\frac{\hat P_E}{g(E)}\right]
\notag\\
&\quad +
\sum_{E\neq E'} 
e^{-\frac{1}{2}(\lambda_E+\lambda_{E'})t}\,
\hat P_E\rho_0 \hat P_{E'},
\end{align}
where $\varepsilon_E(t)=1-e^{-t\gamma_E}$. As $t\to\infty$, $\varepsilon_E(t)\to1$ and $e^{-\frac{1}{2}(\lambda_E+\lambda_{E'})t} \to 0$, yielding the limiting state
\[
\lim_{t\to\infty}\rho(t)
=
\sum_E \mathrm{tr}(\hat P_E\rho_0)\,\frac{\hat P_E}{g(E)}
=\mathcal{T}_{\mathrm{EPU}}(\rho_0).
\]
Thus the dynamics converges exponentially to $\mathcal{T}_{\mathrm{EPU}}(\rho_0)$, the unique fixed point compatible with the conserved shell populations $\mathrm{tr}(\hat P_E\rho_0)$.

\medskip
\noindent\textbf{Finite-time convergence to the EPU twirl.} Let  $\Delta\rho_0 = \rho_0-\mathcal{T}_{\mathrm{EPU}}(\rho_0)$. Since $e^{t\mc{L}}\left(\mathcal{T}_{\mathrm{EPU}}(\rho_0)\right)= \mathcal{T}_{\mathrm{EPU}}(\rho_0)$, we have $e^{t\mc{L}}\left(\Delta\rho_0\right) = e^{t\mc{L}}(\rho_0)-\mathcal{T}_{\mathrm{EPU}}(\rho_0)$. Thus,
\begin{align}
   \tnorm{e^{t\mc{L}}(\rho_0)-\mathcal{T}_{\mathrm{EPU}}(\rho_0)} =\tnorm{e^{t\mc{L}}\left(\Delta\rho_0\right)}.
\end{align}
Now we bound the term  $\tnorm{e^{t\mc{L}}\left(\Delta\rho_0\right)}$. Note that $\hat{P}_E\rho_0\hat{P}_E-\operatorname{tr}(\hat{P}_E\rho_0)\hat{P}_E/g(E)=\hat{P}_E\Delta\rho_0\hat{P}_E$
and $\hat{P}_E\rho_0\hat{P}_{E'}=\hat{P}_E\Delta\rho_0\hat{P}_{E'}$ for $E\neq E'$. Then
\begin{align}
e^{t\mathcal{L}}(\Delta \rho_0)
=\sum_E e^{-t\gamma_E}\,\hat{P}_E\,\Delta\rho_0\,\hat{P}_E
+\sum_{E\neq E'}e^{-\frac{1}{2}(\lambda_E+\lambda_{E'})t}\,\hat{P}_E\,\Delta\rho_0\,\hat{P}_{E'}.
\label{eq:exact-deviation}
\end{align}
Defining $M:=\sum_E e^{-\lambda_E t/2}\,\hat{P}_E$, we get $M\,\Delta\rho_0\,M
=\sum_{E\neq E'}e^{-\frac{1}{2}(\lambda_E+\lambda_{E'})t}\hat{P}_E\Delta\rho_0\hat{P}_{E'}
+\sum_{E}e^{-\lambda_E t}\,\hat{P}_E\Delta\rho_0\hat{P}_E$. Then
\begin{align}
\tnorm{e^{t\mathcal{L}}(\Delta \rho_0)}\leq \tnorm{M\,\Delta\rho_0\,M}+\sum_E\big|e^{-\gamma_E t}-e^{-\lambda_E t}\big|\,\tnorm{\hat{P}_E\,\Delta\rho_0\,\hat{P}_E} .
\label{eq:sandwich-identity}
\end{align}
Using H\"older's inequality, we have $\lVert M\,\Delta\rho_0\,M\rVert_1\le\lVert M\rVert_\infty^2\,\lVert\Delta\rho_0\rVert_1
\le \max_E\left\{e^{-\lambda_E t}\right\}\,\lVert\Delta\rho_0\rVert_1$. Define $\Gamma:=\min_E\{\gamma_E,\lambda_E\}$, then $\max_E\left\{e^{-\lambda_E t}\right\}\leq e^{-\Gamma t}$. Further, $\left|e^{-\gamma_E t}-e^{-\lambda_E t}\right|\leq e^{-\Gamma t}$ as both  $e^{-\gamma_E t}$ and $e^{-\lambda_E t}$ lie in $[0,e^{-\Gamma t }]$. Using this, we have
\begin{align}
\sum_E \big|e^{-\gamma_E t}-e^{-\lambda_E t}\big|\,\tnorm{\hat{P}_E\Delta\rho_0\hat{P}_E}
&\le e^{-\Gamma t} \sum_E 
\tnorm{\hat{P}_E\Delta\rho_0\hat{P}_E}= e^{-\Gamma t}\tnorm{\sum_E  \hat{P}_E\Delta\rho_0\hat{P}_E}\leq e^{-\Gamma t}\tnorm{\Delta\rho_0},
\end{align}
where in the last line we used the fact that trace-preserving unital channels do not increase trace norm. Using $\tnorm{\Delta\rho_0}\le2$, we get
\begin{align}
\big\lVert e^{t\mathcal{L}}(\rho_0)-\mathcal{T}_{\rm EPU}(\rho_0)\big\rVert_1
\le 2\,e^{-\Gamma t}\,\lVert\Delta\rho_0\rVert_1\le 4\,e^{-\Gamma t}.
\label{eq:finite-time-bound}
\end{align}

If we want $\tnorm{e^{t\mc{L}}(\rho_0)-\mathcal{T}_{\mathrm{EPU}}(\rho_0)}\leq \epsilon$, then choosing $4 e^{-t\Gamma}\leq \epsilon$ suffices. Solving for $t$, we see that any state approaches the EPU-invariant state on a timescale
\(t=O\left(\Gamma^{-1}\log(1/\epsilon)\right)\). Using Theorem~\ref{th:trace-de-finetti} and the contractivity of the trace norm under partial trace,
for any \(k\ll \sqrt{N}\) we obtain
\begin{align*}
\tnorm{
\operatorname{tr}_{N-k}\!\left( e^{t\mc{L}}(\rho_0) \right)
-
\int\mu(\dbeta)\,\tau_{\beta}^{\otimes k}
}
&\le
\tnorm{ e^{t\mc{L}}(\rho_0) - \mathcal{T}_{\mathrm{EPU}}(\rho_0) }
+
\tnorm{
\operatorname{tr}_{N-k}\!\left( \mathcal{T}_{\mathrm{EPU}}(\rho_0) \right)
-
\int\mu(\dbeta)\,\tau_{\beta}^{\otimes k}
}
\\
&\le 4\,e^{-\Gamma t}+\frac{k(8\kappa+d(k+3)-2k+6)}{2N}+ O\left(N^{-3/2+c\delta}\right).
\end{align*}
From above, we see that any state approaches to the EPU-invariant state, and hence from Theorem~\ref{th:trace-de-finetti} to the thermal mixture,  on a timescale
\(t=O\left(\Gamma^{-1}\log(1/\epsilon)\right)\).

\end{document}